%% file: sig-conf-main.tex
\documentclass[sigconf]{acmart}
\input{headers/defs}

\AtBeginDocument{%
  }

\copyrightyear{2025}
\acmYear{2025}
\setcopyright{acmlicensed}
\acmConference[CCS '25] {Proceedings of the 2025 ACM SIGSAC Conference on Computer and Communications Security}{ October 13--17, 2025}{Taipei, Taiwan.}
\acmBooktitle{Proceedings of the 2025 ACM SIGSAC Conference on Computer and Communications Security (CCS '25), October 13--17, 2025, Taipei, Taiwan}
\acmISBN{979-8-4007-1525-9/2025/10}
\acmDOI{10.1145/3719027.3765069}
\begin{document}

\title{\sys: Bandwidth-Optimized Lightning-Fast Oblivious Map powered by Secure HBM Accelerators}

\author{Yitong Guo}
\orcid{0009-0009-8140-9177}
\affiliation{%
  \institution{Indiana University}
  \city{Bloomington}
  \state{Indiana}
  \country{USA}
}
\email{yitoguo@iu.edu}  

\author{Hongbo Chen}
\orcid{0000-0001-9922-4351}
\affiliation{%
  \institution{Indiana University}
  \city{Bloomington}
  \state{Indiana}
  \country{USA}
}
\email{hc50@iu.edu}  

\author{Haobin Hiroki Chen}
\orcid{0009-0002-6888-0721}
\affiliation{%
  \institution{Indiana University}
  \city{Bloomington}
  \state{Indiana}
  \country{USA}
}
\email{haobchen@iu.edu}

\author{Yukui Luo}
\orcid{0000-0002-5852-4195}
\affiliation{%
  \institution{SUNY Binghamton}
  \city{Binghamton}
  \state{New York}
  \country{USA}
}
\email{yluo11@binghamton.edu}

\author{XiaoFeng Wang}
\orcid{0000-0002-0607-4946}
\affiliation{%
  \institution{Nanyang Technological University}
  \country{Singapore}
  }
\email{xiaofeng.wang@ntu.edu.sg}

\author{Chenghong Wang}
\orcid{0000-0001-7837-5791}
\affiliation{%
  \institution{Indiana University}
  \city{Bloomington}
  \state{Indiana}
  \country{USA}
}
\email{cw166@iu.edu}
\renewcommand{\shortauthors}{Guo et al.}

\begin{abstract}

While Trusted Execution Environments provide a strong foundation for secure cloud computing, they remain vulnerable to access pattern leakages. Oblivious Maps (OMAPs) mitigate this by fully hiding access patterns but suffer from high overhead due to randomized remapping and worst-case padding. We argue these costs are not fundamental. Modern accelerators featuring High-Bandwidth Memory (HBM) offer a new opportunity: Vaswani et al. [{\em OSDI '18}] point out that eavesdropping on HBM is difficult—even for physical attackers—as its memory channels are sealed together with processor cores inside the same physical package. Later, Hunt et al. [{\em NSDI '20}] show that, with proper isolation, HBM can be turned into an unobservable region where both data and memory traces are hidden. This motivates a rethink of OMAP design with HBM-backed solutions to finally overcome their traditional performance limits.


Building on these insights, we present \sys, a Bandwidth Optimized, Lightning-fasT OMAP accelerator that, for the first time, achieves $O(1)+O(\log_2\log_2 N)$ bandwidth overhead. \sys introduces three key innovations: (i) a new OMAP algorithm that leverages isolated HBM as an unobservable cache to accelerate oblivious access to large host memory; (ii) a self-hosted architecture that offloads execution and memory control from the host to mitigate CPU-side leakage; and (iii) tailored algorithm-architecture co-designs that maximize resource efficiency.  We implement a prototype \sys on a Xilinx U55C FPGA. Evaluations show that \sys achieves up to \(279\times\) and \(480\times\) speedups in initialization and query time, respectively, over state-of-the-art OMAPs, which includes an industry implementation from Facebook.

\end{abstract}




\maketitle


\input{tex/intro}

\input{tex/background}
\input{tex/model}

\input{tex/complete}

\input{tex/implementation}
\input{tex/evaluation}
\input{tex/related}

\input{tex/conclusion}
\input{tex/acknowledgements}

\bibliographystyle{ACM-Reference-Format}
\bibliography{ref}

\appendix
\input{tex/appendix}

\end{document}

%% file: headers/defs.tex
\usepackage{adjustbox}
\usepackage{multirow}
\usepackage{courier}
\usepackage{dtrt}
\usepackage{tikz}
\usepackage{stmaryrd} 
\usepackage{subcaption}
\usepackage{lineno}
\usepackage{mathtools}
\usepackage{amsmath}
\usepackage{circledsteps} 
\usepackage{tabularx} 
\usepackage{graphicx}
\usepackage{amsfonts}
\usepackage{booktabs}
\usepackage[flushleft]{threeparttable}
\usepackage{pifont}
\usepackage{pgfplots}
\usepackage[stable]{footmisc}
\usepackage{makecell}
\usepackage[utf8]{inputenc}
\usepackage{listings}
\usepackage{semantic}
\usepackage{listings-coq}
\usepackage{listings-rust}
\usepackage{listings-sql}
\usepackage{xspace}
\usepackage{amsthm}
\usepackage[russian,english]{babel}
\usepackage{algorithm}
\usepackage{algorithmicx}
\usepackage{algpseudocode}
\newcommand{\eat}[1]{}
\usepackage{enumitem}
\usepackage{hyperref}
\usepackage{tcolorbox}

\newtheorem{definition}{Definition}[section]

\newtheorem{claim}{Claim}[section]

\definecolor{dark-red}{rgb}{0.4, 0.15, 0.15}
\definecolor{dark-blue}{rgb}{0.15, 0.15, 0.4}
\definecolor{medium-red}{rgb}{0.5, 0, 0}
\definecolor{medium-blue}{rgb}{0, 0, 0.5}
\definecolor{light-red}{rgb}{0.7, 0, 0}
\definecolor{light-blue}{rgb}{0, 0, 0.7}
\definecolor{dark-green}{rgb}{0.0, 0.5, 0.0}

\definecolor{light-gray}{gray}{0.95}
\newcommand{\code}[1]{\colorbox{light-gray}{\texttt{#1}}}

\hypersetup{
  colorlinks,
  linkcolor={dark-red},
  citecolor={dark-blue},
  urlcolor={medium-blue}
}

\newcommand{\re}[1]{\textcolor{black}{#1}}

\newcommand{\sys}{\textsf{BOLT}\xspace}
\newcommand{\sysm}{\textsf{ObliKV-Accel.mem}\xspace}

%% file: tex/intro.tex

\section{Introduction}
\label{sec:intro}
With the rise of cloud computing, ensuring the privacy and security of outsourced data has become increasingly critical. Trusted execution environments (TEEs)~\cite{sgx, amd_sev} have emerged as a powerful solution, offering attestable data-in-use security with far less overhead than cryptographic approaches~\cite{micali1987play, yao1986generate}. However, mainstream CPU TEEs remain vulnerable to side-channel leakages—particularly through memory access patterns~\cite{recorve-via-cache-attack,Kocher2018spectre,Lipp2018meltdown,pessl2016drama,hu2020deepsniffer-2,gross2019breaking, zuo2020sealing, lee2020off}—which can severely undermine their confidentiality guarantees and cause significant privacy breach~\cite{kellaris2016generic, cash2015leakage, blackstone2019revisiting, markatou2019full, kornaropoulos2022leakage, oya2021hiding}.

Oblivious RAM (ORAM)~\cite{goldreich1996software,stefanov2018path} is recognized as the {\em de-facto} solution for mitigating access pattern leakages. In a nutshell, it lets a trusted client dynamically shuffle memory accesses so that each request is served correctly, but the overall access pattern looks completely random. Recent work builds on ORAMs to create Oblivious Maps (OMAPs)~\cite{wang2014oblivious}, which support more advanced in-memory key-value stores (KVSs). Newer designs~\cite{mishra2018oblix,chamani2023graphos,zheng2024h} go a step further by eliminating the need for a trusted client to coordinate execution, which helps cut down communication overhead significantly and makes OMAPs better suited for outsourced computing.


Despite their flexibility and strong security guarantees, OMAPs come with significant performance overheads. OMAPs logically arrange data using search-efficient structures such as AVL trees~\cite{wang2014oblivious, tinoco2023enigmap} or hash tables~\cite{asharov2020optorama,zheng2024h}, and traverse them to find the target KV pair. To maintain obliviousness, however, each accessed item must be randomly remapped to a new position~\cite{oblix,tinoco2023enigmap}—or, after enough accesses, the entire dataset is reshuffled~\cite{zheng2024h}. On top of that, each access is padded with a large number of dummy operations to reach a worst-case access length. While this makes all operations look the same, it often incurs $O(\log_2 N)$ rounds and $O(\log_2^2 N)$ bandwidth blow up per access~\cite{tinoco2023enigmap, zheng2024h}, where $N$ is the number of records. In practice, this can result in more than a 2000$\times$ query slowdown compared to non-private KVSs (\S~\ref{sec:eva}). In fact, any OMAP must pay at least $\Omega(\log N)$ bandwidth overhead~\cite{goldreich1996software}, which fundamentally limits how much performance can be improved.

Given this lower bound, recent research has started exploring a different direction: hardware-unobservable memory (HUM)~\cite{xu2019hermetic, aga2017invisimem, oh2020trustore, vasiliadis2014pixelvault, awad2017obfusmem}, a stronger form of isolated memory designed to hide both data contents and access patterns at the hardware level. HUMs are typically built using on-chip memory, such as caches or UltraRAM~\cite{amd_ultraram}. Unlike DRAM, which connects to the processor via exposed copper traces, on-chip memory sits directly on the processor's silicon die. This tight integration makes it much harder for attackers to snoop on memory traces, even with physical access (assuming proper isolations). The catch, however, is that its capacity is still limited and cannot support scalable data.

Fortunately, recent hardware trends point to a promising alternative: high-bandwidth memory (HBM), now widely adopted in modern accelerators such as GPUs~\cite{nvidia_confidential_computing_h100}, FPGAs~\cite{AMD_Alveo_U55C}, and ASICs~\cite{amd-mi325,dhar2024ascend}.  Like on-chip memory, HBM is packaged together with processor dies and communicates via silicon interposers, but it provides much larger capacity—ranging from 16GB~\cite{AMD_Alveo_U55C} to 256GB~\cite{amd-mi325}. Prior work on accelerator TEEs~\cite{hunt2020telekine,volos2018graviton} suggests that, with proper isolation, HBM can also be treated as a form of HUM. This leads to the central question of this paper:
\begin{center} {\em Can HBM-backed HUM lead to new secure KVS solutions that match OMAP’s security but incur a much lower overhead?} \end{center} 
\vspace{-1em}

\subsection{Challenges and Key Ideas}\label{sec:challenges}
\noindent{\bf C-1. Bounded HBM capacity.} Although HBM offers much more capacity than on-chip memory, it remains insufficient for modern cloud and data center workloads that may require large in-memory data. Hence, simply using HBM as a secure memory extension~\cite{oh2020trustore, choi2024shieldcxl} would not suffice. To scale, we need new approaches that go beyond HBM's capacity while preserving obliviousness.

\vspace{2pt}\noindent {\em Key ideas.} We propose a novel heterogeneous layout: the main data resides in host memory and is accessed through oblivious primitives, while HBM is used to store metadata and run rich, data-dependent algorithms that accelerate oblivious access to host memory. This idea is motivated by the observation that much of the overhead in existing OMAPs stems from the assumption that only a constant-sized private memory is available—whether in the form of HUM or trusted client storage—forcing designs to rely heavily on exhaustive padding and obfuscation. HBM, however, can scale proportional to the host memory capacity\footnote{Currently, the largest production HBM is 256G by AMD MI325~\cite{amd-mi325} and the cap DIMM capacity for a single-socket CPU is 6TB~\cite{AMD_EPYC_9004_8004}.}, which allows us to offload enough data-dependent operations and to simplify oblivious primitives.

\vspace{4pt}\noindent{\bf C-2. Indirect leakage from the host.} Hunt et al.~\cite{hunt2020telekine} show that even when HBM is isolated within a GPU TEE, attackers can still exploit indirect leakages via the host CPU to recover sensitive results inside the TEE. This is because modern accelerators continue to rely on host-side drivers for tasks such as I/O control, memory management, and execution dispatch. Consequently, critical accelerator states—like HBM accesses and control flows—remain vulnerable to CPU interference and its microarchitectural weakness. Hunt et al. argue that CPU side-channels, even with TEEs, make rigorous secure design incredibly hard. So they propose to move these host features to a trusted client. However, this approach requires consistent client involvement and can impose a significant communication overhead, especially for OMAP routines that would need frequent I/Os and dynamic HBM management. 





\vspace{2pt}\noindent {\em Key ideas.} Instead of involving a trusted client, we take a different path—we move key host-side features into the accelerator architecture and behind isolation boundaries. The accelerator self-manages KV logic, dynamic memory, and oblivious access primitives on its own while exposing only high-level interfaces to the outside. A minimal runtime is then left on the host side, used only for initialization and message relaying. This approach is similar to host bypassing in the HPC community~\cite{nvidia_gpudirect, guo2016rdma} for performance, we repurpose it to minimize host involvement and reduce leakage.

\vspace{4pt}\noindent{\bf C-3. Hardware inflexibility.} Clearly, achieving these ideas requires new architectural designs. However, since hardware is less flexible in design than software, poor design choices can easily lead to large resource fragmentation or inefficient data flow. Moreover, adding a generic software layer on top often results in bloated logic and additional overhead and may require extra effort on compilers.

\vspace{2pt}\noindent {\em Key ideas.} We choose to stay on a hardware-based solution, but carefully navigate co-design optimizations across both the OMAP algorithms and the hardware logic. This will lead to a customized architecture that is deeply optimized for OMAP tasks.

\subsection{Our Outcomes}
Building on our key ideas, we present \sys, a Bandwidth-Optimized Lightning-fast OMAP accelerator that breaks through the long-standing performance limits of traditional OMAPs. \sys is the first known design to achieve $O(1) + O(\log_2 \log_2 N)$ bandwidth overhead, enabling ultra-efficient secure KVS with full obliviousness. \autoref{fig:general} provides an overview of \sys design.
\begin{figure}[ht]
\centering
\includegraphics[width=0.8\linewidth,interpolate=false]{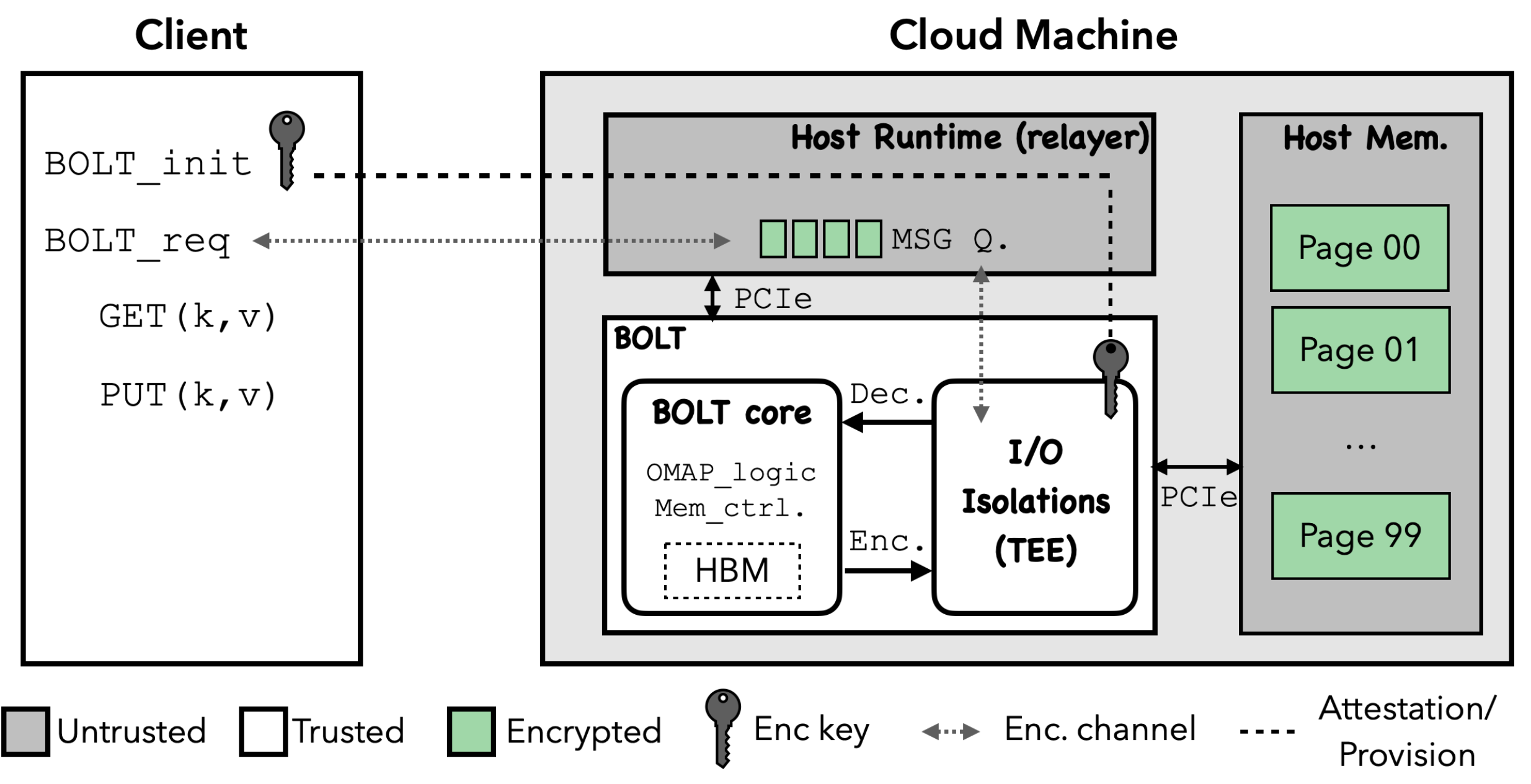}
\caption{\sys overview and deployments}
\label{fig:general}
\end{figure}
At a high level, \sys is built on top of accelerator TEE architectures. It uses existing TEE features (§\ref{sec:fpga}) to enforce physical isolation and ensure integrity. The \sys core is a full-fledged OMAP engine that sits behind the TEE gateway to provide efficient oblivious KV accesses. A typical \sys lifecycle goes like this: A client first goes through standard TEE setups, such as remotely attesting the device~\cite{zhao2022shef} and securely provisioning it with a secret key~\cite{oh2020trustore}. The key is used to encrypt all communication with \sys and allows it to safely seal data to host DRAMs. Once set up, the client sends encrypted KV commands (e.g., $\texttt{GET}$ and $\texttt{PUT}$) to read or write data. These requests are handled by the \sys core, which then returns fixed-length encrypted responses. All I/O to and from \sys core is funneled through the upfront TEE interfaces~\cite{zhao2022shef}, so any sensitive content, such as commands, data, or responses, remains encrypted whenever outside the isolation boundary.

To realize our key ideas and address the aforementioned challenges, \sys introduces the following non-trivial contributions.
\begin{itemize}[leftmargin=.1in]
\item To tackle {\bf C-1}, we introduce a new OMAP scheme (\S~\ref{sec:omap-alg}) that strategically uses unobservable HBM and host DRAM to reach a major performance boost. The core idea is to logically shuffle data into fixed-size bins across both memory types. Each access reads an entire bin, after which the data is randomly remapped to new bins to maintain obliviousness. A hash table stored in HBM keeps track of data-to-bin mappings, enabling fast, leakage-free lookups. Moreover, we novelly apply the power-of-two-choices (P2C) technique~\cite{mitzenmacher2001power, sitaraman2001power} to balance bin loads, which is central to achieving the \smash{$O(1) + O(\log_2 \log_2 N)$} bandwidth overhead.\vspace{1mm}


\item We perform rigorous dimensional analysis (\S~\ref{sec:sizes}) to derive analytical upper bounds for key data structures used in our logical algorithm. These upper bounds guide parameter choices and memory allocation to prevent overflows. We later validate their tightness through empirical experiments. \vspace{1mm}

\item To address {\bf C-2}, we design a custom hardware architecture that realizes our logical algorithm and operates self-contained (\S~\ref{sec:arch-details}). It streams all OMAP routines and manages both HBM and host data by itself. This eliminates the need for a trusted client and mitigates indirect leakage from host CPUs. To address {\bf C-3}, we explore several co-design optimizations (\S~\ref{sec:optimization}), including decomposed storage, reverse indexing, and dynamic HBM management. \vspace{1mm}

\item We prototype \sys on a Xilinx U55C FPGA and benchmark it against several state-of-the-art (SOTA) OMAPs, including an industry-grade solution from Facebook~\cite{facebook_oram}. Results show that \sys achieves up to \(279\times\) and \(480\times\) lower init and query times, respectively, compared to SOTA designs, which corresponds to \(760\times\) and \(6338\times\) improvements in normalized slowdown.

\end{itemize}


%% file: tex/background.tex
\section{Background}\label{sec:background}
\subsection{General Notations of KVS}
We model $D$ as a key-value store (KVS) database, where $D = {\langle k_i, v_i \rangle}_{1 \leq i \leq N}$ and each $\langle k_i, v_i \rangle$ is a key-value pair, with $k_i$ as the unique key for its value $v_i$. We consider two operations on $D$: $\texttt{Get}(k_i, D)$ retrieves the value for key $k_i$, returning $v_i$ or $\mathsf{null}$ if the key does not exist; and $\texttt{Put}(k_i, v_i, D)$ inserts or updates a key-value pair. Deletions are represented as special $\texttt{Put}$ operations where $v_i$ is a tombstone marker, after which the pair is removed from $D$. For simplicity, we omit $D$ when referring to these operations.

\subsection{Access Patterns, ORAMs, and OMAPs}\label{subsec:omaps}
\noindent{\bf Access patterns.} Access pattern leakage is a major side-channel threat in secure processors. It refers to the sequence in which a user program accesses memory. Since these patterns often depend on sensitive data or secret-dependent branches, they can leak critical information. Attacks leveraging such leakages typically fall into two categories. The first exploits {\em shared microarchitectural resources}—such as branch predictors~\cite{Kocher2018spectre, branch_predictor, 2-level-predictor-side-channel-attack, branch-shadowing}, caches~\cite{liu2015last, cache-side-channel-leak-key, yarom2014flush+}, and TLBs~\cite{gras2018tlbleed,tatar2022tlb}—to infer victim's secrets based on its access patterns. The second targets {\em unprotected memory channels} like DRAM buses~\cite{pessl2016drama}, PCIe~\cite{hu2020deepsniffer-2}, DMA buffers~\cite{gross2019breaking}, and unsealed copper traces that connect memory~\cite{zuo2020sealing}.




\vspace{4pt}\noindent{\bf OMAPs.} An OMAP~\cite{omap-data-structure, tinoco2023enigmap, oblix, chamani2023graphos, zheng2024h, wang2014oblivious} is a cryptographic primitive that allows a client to interact with an outsourced KVS on an untrusted server without revealing the access pattern or data content. OMAPs provide foundational storage for more complex oblivious computations, including analytical queries~\cite{oblidb, chang2021efficient} and transactional workloads~\cite{crooks2018obladi, dauterman2021snoopy} over encrypted databases. Most OMAPs are built on top of ORAMs~\cite{goldreich1987towards, goldreich1996software, stefanov2018path}, which were originally designed to hide access patterns in the standard RAM model. In this model, memory is represented as a sequence of address-value pairs 
\(\{(\mathsf{addr}_i, v_i)\}_{i=0}^{n-1}\), and the addresses are a unique, consecutive integers. ORAM supports reading and writing to these pairs while making access patterns appear random. However, an ORAM does not directly yield an OMAP, as OMAPs must support arbitrary unique\footnote{Multi-maps with non-unique keys are beyond our scope.} keys, such as strings or sparse indexes.

SOTA OMAP designs today fall into two main categories: (i) projects~\cite{oblix, chamani2023graphos, cao2024towards, tinoco2023enigmap, facebook_oram, wang2014oblivious} built on oblivious data structures layered over tree ORAMs~\cite{path-oram, ren2015constants, resizable-tree-based-oram, bindschaedler2015practicing, sasy2017zerotrace, fletcher2015low}; and (ii) designs~\cite{zheng2024h} adapted from hierarchical hash-based ORAMs~\cite{asharov2023futorama, asharov2020optorama, patel2018panorama, dittmer2020oblivious} for supporting arbitrary keys.

Tree ORAM organizes memory as a binary tree, where each node (or bucket) holds a fixed number of data blocks, and each block is mapped to a random leaf. To hide access patterns, it uses path-based access and remapping~\cite{path-oram}: every access reads all buckets in the path from the root to the block’s assigned leaf, then remaps the block to a new random leaf. To prevent leakage, the block is not written back directly but is kept in a trusted stash and later evicted when a future access hits an overlapping path. 
This design incurs an $O(\log_2 N)$ bandwidth overhead when the client manages the full position map~\cite{path-oram}. To turn this into an OMAP, Wang et al.~\cite{wang2014oblivious} propose storing KV data as a logical AVL tree within the tree ORAM. The bandwidth overhead per KV access becomes $O(\log_2^2 N)$.

Hierarchical hash ORAMs, derived from square-root ORAMs~\cite{goldreich1987towards}, organize memory into $O(\log_2 N)$ levels of increasingly large hash tables. To access a block, the client scans each non-empty level, retrieves the block, and writes it back to the first level. After every $2^i$ accesses, data from the first $i$ levels is merged and rebuilt into level $i+1$ (the most expensive step). These ORAMs naturally support KV search due to their hash table structure and achieve the optimal $O(\log_2 N)$ bandwidth, though large constant factors limit their practical performance~\cite{cao2024towards, oblix}.



\vspace{4pt}\noindent{\bf Client-server vs. outsourcing paradigm.} Traditional ORAMs follow a client-server model~\cite{path-oram, ren2015constants}, where a trusted client handles the control logic, including data shuffling, path remapping, position map lookups, stashing, and evictions. The server, by contrast, simply stores the outsourced data and serves requests. This design requires the client program to stay online and actively involved in every access, leading to high client-side overhead and significant communication costs. Recent OMAPs often use TEEs as a trusted controller, simplifying client-side tasks and reducing communication overhead. However, CPU-based TEEs do not naturally hide memory access patterns and may still leak sensitive information, for example through data-dependent position lookups or stash management. To address this, Mishra et al.~\cite{mishra2018oblix} introduced doubly-obliviousness (DO), which requires hiding the memory access patterns of both the data and the control program. Achieving DO typically comes at a cost: position maps may need to be stored in recursive ORAMs~\cite{stefanov2018path}, and eviction must be handled with branchless operations or specialized algorithms~\cite{oblix, chamani2023graphos, oblidb}. Hence, current SOTA DOMAPs typically incur $O(\log_2 N)$ rounds and $O(\log_2^2 N)$ bandwidth overhead. Still, by eliminating client-server interaction, this overhead is limited to memory access rather than network communication, which generally results in better practical performance~\cite{oblix, chamani2023graphos, zheng2024h}. We adopt this outsourcing paradigm with DO as our default model and refer to it simply as obliviousness.


\subsection{Accelerators}\label{sec:fpga}
Modern HPC has shifted from a CPU-centric model to a heterogeneous architecture, where specialized accelerators—such as GPUs, FPGAs, and ASICs—offload intensive workloads from CPUs. Unlike CPUs, which prioritize flexibility and time-sharing, accelerators dedicate resources to specific tasks with a minimal software stack for near-bare-metal performance. 
In this work, we innovatively explore offloading OMAP tasks to secure HBM accelerators.

\vspace{3pt}\noindent{\bf Accelerator TEEs.} Accelerator vendors have traditionally offered security features like hardware Root of Trust (RoT)~\cite{amd_hwrort_2024, nvidia_confidential_computing_h100, dhar2024ascend}, bitstream encryption (for FPGAs), and secure boot to verify device authenticity and enforce integrity at boot time. Recent works~\cite{zhao2022shef,armanuzzaman2022byotee,volos2018graviton,hunt2020telekine, mai2023honeycomb, wang2024towards,jang2019heterogeneous, ivanov2023sage,nvidia_confidential_computing_h100, vaswani2022confidential, dhar2024ascend,oh2020trustore} extend these features to support richer TEE capabilities, including: (i) {\em Remote attestations.} A remote party can verify the fabric configuration (e.g., FPGAs, ASICs) and trusted firmware (e.g., GPUs, DPUs), and ensures the accelerator maintains a secure runtime state~\cite{zhao2022shef, oh2020trustore, nvidia_confidential_computing_h100, dhar2024ascend,ivanov2023sage}. (ii) {\em Isolated executions.} Hardware firewalls and access controls are used to create physically isolated enclaves on accelerators or turn the entire accelerator into a standalone, secure device~\cite{armanuzzaman2022byotee, nvidia_confidential_computing_h100, zhao2022shef, dhar2024ascend, volos2018graviton, hunt2020telekine}. During execution, the enclave is protected from external access, tampering, or interruption. Device I/Os like memory-mapped I/Os (MMIO) and direct memory accesses (DMAs) are typically funneled through an isolation gateway~\cite{nvidia_confidential_computing_h100, zhao2022shef, oh2020trustore}, which wraps these I/Os and encrypts outbound data while decrypting inbound data using keys inside the enclave. Performance counters are often disabled to prevent unauthorized information leakage~\cite{zhao2022shef,nvidia_confidential_computing_h100}.  (ii) {\em Memory encryption and integrity.} Data can be sealed outside the TEE, e.g. in non-secure host memory, but is encrypted~\cite{gueron2017aes} and integrity validated~\cite{liu2021merkle}. Encryption keys are either generated within enclaves~\cite{nvidia_confidential_computing_h100, dhar2024ascend} or securely provisioned by users~\cite{zhao2022shef}. While these TEE features are essential building blocks for our end-to-end OMAP accelerator, our design assumes their availability and does not contribute new mechanisms in this space. Since these techniques are well-established and orthogonal to our core contributions, we omit their details for brevity and focus instead on the novel aspects of our OMAP logic. Readers interested in these TEE components can refer to Appendix~\S\ref{appx:acc-tee} for additional background.

\vspace{4pt}\noindent{\bf HBMs.} HBM is an advanced memory technology that stacks multiple DRAM layers using 3D-integrated circuit manufacturing, achieving higher bandwidth through dedicated data channels~\cite{Lee2022}. Unlike traditional off-chip DRAM, which connects to processors via exposed PCB traces, HBM is tightly integrated with compute cores inside the same chip package using silicon interposers. These sealed interposers shield memory channels, making direct snooping or tampering extremely difficult when proper isolation mechanisms are in place~\cite{volos2018graviton, hunt2020telekine}. However, this physical protection alone does not make HBM oblivious. For example, in HBM-based CPUs~\cite{intel_xeon_max}, memory remains shared across cores and programs, making it vulnerable to microarchitectural attacks such as cache side channels. Hunt et al.\cite{hunt2020telekine} also show that isolated HBM inside GPU TEEs can suffer indirect data leaks via host CPUs (or their TEEs), since today's accelerators heavily depend on the host runtime for memory and execution management. We tackle this limitation in \S\ref{sec:impl}.

\vspace{4pt}\noindent\re{{\bf FPGA prototyping.} In this paper, we focus on FPGA prototyping to implement and evaluate our design for two main reasons. First, many modern cloud and datacenter infrastructures already integrate FPGA accelerators~\cite{AWSF1_2017, AWSF2_2024}, enabling our design to be directly deployed as an independent secure KVS accelerator that also aligns with the resource disaggregation paradigm. Second, FPGAs offer greater architectural flexibility and are a standard pre-production step for ASICs, while also supporting extensions of fixed-function accelerators like GPUs and DPUs. For example, GPU TEEs could use \sys as a secure data fetcher to self-manage data I/O.}

%% file: tex/model.tex
\section{Threat Model and Design Goals}\label{sec:threat}

\noindent{\bf Threat model.} Our work follows the standard secure outsourced computing model with two primary entities: the cloud service provider (CSP) and the data owner. The CSP supplies and manages the infrastructure, including the proposed \sys accelerators, to host confidential KVS services. The client wishes to securely outsource a private KVS database to the CSP and access it as needed. Our threat model assumes trust in the CSP’s organizational integrity and governance, but distrusts lower-level components such as software, operational personnel, and co-located users. Specifically, we consider a strong adversary capable of compromising any software stack and gaining physical access to hardware, enabling stealthy (passive) physical attacks such as bus snooping. In general, we assume that all exposed links and buses—including DMA~\cite{gross2019breaking}, host DRAM~\cite{pessl2016drama}, device memory buses (e.g., DDR on FPGA boards~\cite{zuo2020sealing}), and PCIe interconnects~\cite{hu2020deepsniffer-2}—are susceptible to snooping. However, we assume attackers cannot perform hypothetical chip depackaging to compromise silicon interposers~\cite{nvidia_confidential_computing_h100} within the chip package. Additionally, each \sys accelerator instance is dedicated to a single tenant, so attacks requiring sophisticated multi-tenancy are out of scope~\cite{zhao2018fpga, fang2023gotcha}.

\vspace{4pt}\noindent{\bf Privacy goals (obliviousness).} Our primary privacy goal is to protect the owner's private data against the aforementioned adversary and deliver strong obliviousness for outsourced KVSs. Specifically, for any data $D$ and a sequence of KV commands, $\mathbf{r}\gets \{c_1, c_2, ...\}$, the information an adversary can learn by observing the outsourcing of $D$ and processing $\mathbf{c}$ over outsourced $D$ should not be better than some public (non-private) information. Formally,

\begin{definition}\label{def:privacy}
For any $D$, $\mathbf{r}$, and any probabilistic-polynomial time (p.p.t.) adversary $\mathcal{A}$, we define $\mathsf{View}_{\mathcal{A}}^{\mathsf{Real}}$ as the view of $\mathcal{A}$ when interacting with the actual secure outsourced KVS system. We say that the system is an oblivious KVS (securely simulates an oblivious KVS design) if there exists a p.p.t. simulator $\mathcal{S}$ that can simulate indistinguishable transcripts as $\mathsf{View}_{\mathcal{A}}^{\mathsf{Real}}$ without access to $D$ and $\mathbf{r}$, or equivalently, if the following holds:
\begin{equation}
\mathsf{View}_{\mathcal{A}}^{\mathsf{Real}} (D, \mathbf{c}) \approx_{\mathsf{ind}} \mathsf{View}_{\mathcal{A}}^{\mathcal{S} (\mathsf{pp})}
\end{equation}
where $\approx_{\mathsf{ind}}$ denotes computational indistinguishability and $\mathsf{pp}$ is a set of public (non-private) parameters, such as $|\mathbf{c}|$ and $|D|$.
\end{definition}


\vspace{4pt}\noindent{\bf Non-goals.} We emphasize that in this work, we do not consider physical channel analysis attacks such as those based on energy~\cite{xiang2020open, li2022power, moini2021remote, yoshida2021model} or electromagnetic emanations~\cite{gongye2024side, batina2019csi}, as these attacks are typically designed for edge devices and are less feasible in well-governed cloud environments. We also exclude availability attacks (e.g., denial-of-service~\cite{luo2020stealthy}) and covert-channel attacks~\cite{tian2019temporal_2, giechaskiel2020c, giechaskiel2022cross}, as they fall outside the general security goals of secure computations and can be independently addressed via orthogonal security measures. We stress that this general exclusion aligns with previous works on secure and oblivious computations~\cite{nvidia2023hcc, xu2019hermetic, volos2018graviton, hunt2020telekine, oh2020trustore, duy2022se}. 
Moreover, several important building blocks, such as remote attestation of FPGA kernels~\cite{oh2020trustore, zhao2022shef}, device I/O isolation~\cite{oh2020trustore, nvidia2023hcc}, memory encryption~\cite{AMD_Vitis_Security}, and integrity validations are related to our design. As there are existing lines of work addressing these features, we leverage them rather than replicating the designs ourselves.


%% file: tex/complete.tex
\section{Logical Algorithm}\label{sec:omap-alg}
We address challenge {\bf C-1} by presenting the logical algorithm for a novel OMAP scheme that utilizes limited HBM space to accelerate oblivious KV accesses for large in-memory data. 

\subsection{Algorithm details}\label{sec:alg-details}
The logical algorithm for our proposed OMAP scheme is intentionally simple, comprising only 15 lines of pseudocode, as shown in Algorithm~\ref{alg:omap}. As the algorithm handles data stored in different locations, we use \textcolor{blue}{blue} text in Algorithm~\ref{alg:omap} to highlight objects stored in host memory. These objects are encrypted, but their access patterns remain visible to attackers. The remaining objects, including the algorithm logic, reside within the accelerator’s on-package resources (e.g., HBM), ensuring that reads, writes, and intermediate runtime states remain unobservable.


\vspace{3mm}
\begin{algorithm}[]
\caption{Logical \sys algorithm}
\begin{algorithmic}[1]
\Statex {\bf Inputs:} HBM store ${V}_{\mathsf{hbm}}[K]$; host store $\textcolor{blue}{{V}_{\mathsf{host}}[M]}$; stash ${V}_{\mathsf{st}}[M]$, position map $\mathsf{MAP}_{\mathsf{p}}$; command $(\textcolor{blue}{\textit{opcode}}, \textcolor{blue}{\textit{key}}, \textcolor{blue}{\textit{payload}})$.
\State $(op, k, ld) \gets \text{load}(\textcolor{blue}{\textit{opcode}}, \textcolor{blue}{\textit{key}}, \textcolor{blue}{\textit{payload}})$
\If{$ \left(p_1, p_2 \in [1, M+K] \gets \text{lookup}({\mathsf{MAP}_{\mathsf{p}}}, k)\right) = \emptyset$}
\Statex \textcolor{gray}{~~~~~~\it //position map miss, insert new data.}
\State $v \gets ld$, \text{dummy\_accesses\_and\_jump\_to(\ref{line:jmp})} 
\EndIf
\For{$p_i \in (p_1, p_2)$}
\If{$p_i \in (0, K]$}  $v \gets \text{find\_remove}({V}_{\mathsf{hbm}}[p_i], k)$
\Else 
        \State $\textit{page} \gets \text{read\_page}(\textcolor{blue}{V_{\mathsf{host}}}[p_i - K])$
        \State $v \gets \text{find\_remove}(page \cup {V}_{\mathsf{st}}[p_i -K], k)$
        \State $\text{write\_back:} ~\textit{page} \cup {V}_{\mathsf{st}} [p_i - K]\rightarrow  \textcolor{blue}{{V}_{\mathsf{host}}}[p_i - K]$
\EndIf
\EndFor
\State $\text{exec\_cmd}(op, k, v, ld)$
\State \label{line:jmp} $p'_1, p'_2 \in [1, M+K]\gets \text{random\_remap()}$
\State {\text{P2C\_load\_balance}($p'_1$, $p'_2$, $\mathsf{MAP}_{\mathsf{p}}$, ${V}_{\mathsf{hbm}}$, ${V}_{\mathsf{st}}$, $k$, $v$)}
\end{algorithmic}
\label{alg:omap}
\end{algorithm}
\vspace{2mm}

Alg.\ref{alg:omap} adopts similar concepts to tree ORAMs (e.g., Path ORAM) that use access-then-remapping mechanisms\cite{pathoram}. However, it simplifies these ideas by removing the tree structure and instead using a bin-level design \re{(or equivalently, a flat single-layer tree)}. At a high level, we consider the entire data store to be divided into \(K+M\) logical bins, where the first \(K\) bins are in HBM (\(V_{\mathsf{hbm}}\)), and the remaining \(M\) are instantiated as fixed-sized, encrypted pages in host memory (\(V_{\mathsf{host}}\)). The core idea of our algorithm is to map real data accesses into two random accesses across logical bins. Each data item is initially assigned to two uniformly chosen random bins, $p_1,p_2$, and placed in one of them. To serve a request (Alg~\ref{alg:omap}:2), the algorithm queries the global position map  \(\mathsf{MAP}_{\mathsf{p}}\) to retrieve $p_1, p_2$ for a given key.  It then accesses both bins concurrently—one fetches the actual data, while the other is dummy to ensure obliviousness. After each access, the data is remapped to two new random bins (Alg~\ref{alg:omap}:14). A key novelty of our scheme is the integration of P2C load balancing (Alg~\ref{alg:omap}:15), where, both in initialization or after remapping, the real data is always placed in the less occupied of the two bins. This feature results a compact worst-case bin size which serves as a key property that leads to the \( O(1) + O(\log_2\log_2 N) \) bandwidth blowup in our accelerator design (\S~\ref{sec:impl}).

When the final destination (after P2C) of a remapped data is an HBM bin, it can be directly inserted into the target bin. However, when the destination is a host page, directly writing the data to the mapped page would leak access patterns~\cite{path-oram}. Thus, we adopt a strategy similar to Path ORAMs, using an eviction stash \(V_{\mathsf{st}}\) (initially empty) to temporarily buffer data evicted from \(V_{\mathsf{hbm}}\) while it is pending write-back to \(V_{\mathsf{host}}\). The actual eviction occurs when a host page read is triggered by a future access. At that point, all data in \(V_{\mathsf{st}}\) are mapped to the same page as the one just read is written back (and re-encrypted) together (Alg~\ref{alg:omap}:10). 
Note that a data item may not be found in the read pages, as it could reside in \(V_{\mathsf{st}}\). Therefore, both \(V_{\mathsf{host}}\) and \(V_{\mathsf{st}}\) must be searched (Alg~\ref{alg:omap}:9).

Once the requested data $(k,v)$ is accessed, the algorithm executes the command based on the request opcode. We consider a standard KVS interface with two opcodes: \texttt{GET} and \texttt{PUT}. For \texttt{GET}, the algorithm returns the retrieved data (in ciphertext). For \texttt{PUT}, it updates $v$ with a given payload or removes $(k,v)$ if the payload includes a tombstone marker. A special case is when $k$ is not found in $\mathsf{MAP}_{\mathsf{p}}$, which suggests an insertion. The algorithm will perform a dummy value access using two random $p_1, p_2$, and proceeds directly to the random remapping phase. The response of \texttt{PUT} is the same size as \texttt{GET} but contains only a confirmation code.

\subsection{Security Analysis}\label{sec:ag-analysis} 

\begin{claim}[Obliviousness]\label{claim:obliv} The logical \sys algorithm defined in Alg~\ref{alg:omap} is data-oblivious (or satisfies Definition~\ref{def:privacy}).
\end{claim}

\begin{proof} We first characterize the transcripts observable by the adversary. Recall that internal states and accesses to on-package resources are unobservable, so only interactions beyond this boundary are visible.
Hence, given \( D \) and \(\mathbf{c} = \{c_1,\ldots,c_n\}\), the adversary observes only: (i) the encrypted commands \( \mathsf{in}_i(c_i, D) \), (ii) the encrypted responses \( \mathsf{out}_i(D, c_i) \), and (iii) the off-package memory accesses \( \mathsf{mem}_i(D, c_i) \). Formally, the adversary's view is
\[
\mathsf{View}_{\mathcal{A}}^{\mathsf{Real}}(D, \mathbf{c}) = \Bigl\{ \bigl( \mathsf{in}_i(c_i, D),\, \mathsf{out}_i(D, c_i),\, \mathsf{mem}_i(D, c_i) \bigr) \Bigr\}_{i=1}^n.
\]
Obliviousness holds if there exists a simulator \( S \) that, using only non-private information (e.g., \( |D| \) and \( |\mathbf{c}| \)), produces a transcript indistinguishable from \( \mathsf{View}_{\mathcal{A}}^{\mathsf{Real}}(D, \mathbf{c}) \). Since all inputs and outputs are encrypted and each operation executes in constant time, the I/O traffic is trivial to simulate; we thus focus on the off-package memory accesses. Let $\mathcal{B} = \{1, 2, \ldots, K+M\}$ denote the set of logical bins. Initially, every key is assigned uniformly at random to two distinct bins. Hence, for any ordered pair \((b_1,b_2) \in \mathcal{B}\times\mathcal{B}\) with \( b_1 \neq b_2 \), when a key \( k \) is accessed for the first time, the probability \smash{$\Pr\Bigl[(b_1,b_2) \text{ is accessed}\Bigr]$} is \smash{$\frac{1}{(K+M)(K+M-1)}$}.
Moreover, for each subsequent access (indexed by a counter \( j \)), the algorithm reassigns \( k \) to two distinct bins using a mapping $\pi: \mathcal{K} \times \mathbb{N} \to \{(b_1,b_2) \in \mathcal{B}\times\mathcal{B} : b_1 \neq b_2\}$, 
so that for any fixed \((b_1,b_2)\) with \( b_1 \neq b_2 \), we have \smash{$\Pr\Bigl[\pi(k,j) = (b_1,b_2)\Bigr] = \frac{1}{(K+M)(K+M-1)}$}.
Thus, every key access—whether the first or a subsequent one—is statistically equivalent to a random access to two logical bins. With this analysis, we now construct a simulator as follows.

\begin{center}
\fbox{%
\begin{minipage}{0.95\linewidth}
\textbf{Simulator \( S(K, M, sz, |\mathbf{c}|) \):}
\begin{enumerate}
    \item \textbf{Init:} Internally simulate \(K+M\) dummy bins and encrypts the $M$ host bins into pages of size $\textit{sz}$.
    \item For each index \( i \in \{1,\ldots,n\} \):
    \begin{enumerate}
        \item Generate a random command \( c'_i \) with a dummy key \( k_i \).
        \item Random a pair $\{ (b_1, b_2) \in \mathcal{B}^2 : b_1 \neq b_2 \}$
        \item For each bin \( b \) in the pair \(\{b_1, b_2\}\):
        \begin{enumerate}
            \item If \( b \le K \) (i.e., \( b \) is in HBM), idle.
            \item Otherwise, simulate  $e_{\mathsf{mem}, i}=(P^{\text{read}}_b, P^{\text{write}}_b)$:
            \begin{enumerate}
                \item Read a random encrypted page \( P^{\text{read}}_b \).
                \item Generate a random ciphertext \( P^{\text{write}}_b \) of the same size to simulate a page writeback.
            \end{enumerate}
        \end{enumerate}
        \item Generate random ciphertexts \( e_{\text{in},i} \) and \( e_{\text{out},i} \).
        \item Output:  $(e_{\text{in},i}, e_{\text{out},i},  e_{\mathsf{mem}, i}=(P^{\text{read}}_b, P^{\text{write}}_b) )$.
    \end{enumerate}
\end{enumerate}
\end{minipage}%
}
\end{center}

Because the real memory accesses are distributed uniformly over the pairs of distinct bins, and the encryption renders inputs, outputs, and memory pages indistinguishable from random data, we have
\[
\mathsf{View}_{\mathcal{A}}^{\mathsf{Real}}(D, \mathbf{c}) \approx_{\mathsf{ind}} \mathsf{View}_{\mathcal{A}}^{S(K, M, sz, |\mathbf{c}|)}.
\]
\end{proof}
\re{In summary, Algorithm~\ref{alg:omap} ensures that each data item is randomly mapped to two bins during either initialization and after every access. Hence, each access in any sequence appears identical—reading two random bins and writing them back. Moreover, evictions are hidden within random write-backs and remain undetectable. Together, these design choices ensure strong obliviousness.}

\subsection{Dimensional analysis}\label{sec:sizes} In this section, we analyze the sizes of several key objects in our logical algorithm, focusing on deriving high-probability upper bounds. These bounds guide memory allocation to prevent overflows, characterize capacity limits (e.g., estimating minimal HBM requirements), and serve as key tools for our subsequent overhead analysis (\S~\ref{sec:impl}). For simplicity, all analyses assume an input data of size $N$, and tolerate a small failure probability of at most \smash{$\frac{1}{O(N)}$}. Note that, when $N$ is large, such as proportional in $2^{k}$, this probability becomes exponentially small. \re{While there is a small chance of overflow causing data loss, this only affects durability guarantees. Even commercial products like AWS S3~\cite{awsS3} do not ensure deterministic durability, so we consider an exponentially small risk of data loss is acceptable.
} 



\begin{claim}[bin load]\label{claim:bin-sz}  Give $N=c(K+M)$, where $c$ is some constant. Then with probability at least \smash{$1-\frac{1}{O(N)}$}, the max load of all bins is bounded by $\ell_{\max}=c + O({\log_2\log_2 N})$
\end{claim}

\begin{proof} The proof of this claim is a direct application of the P2C theorem~\cite{mitzenmacher2001power, sitaraman2001power}. For brevity, we do not repeat the proof details here but provide the full proof in \S~\ref{appx:p2c} for completeness.
\end{proof}

Since bin sizes are strictly bounded by $\ell_{\max}$, fixing the page size to $\ell_{\max}$ suffices to prevent page overflows. This holds because, with the presence of the eviction stash, the page size is at most equal to the corresponding logical bin size. Trivially, one can also derive an upper bound on the size of $V_{\mathsf{hbm}}$ as $K\ell_{\max}$. Nevertheless, the above bound may be overly pessimistic. Since $V_{\mathsf{hbm}}$ reside within the HBM and is not observable by an attacker, we can employ dynamically sized bins instead of fixed-size pages. Hence, we need to derive a tighter upper bound.


\begin{claim}[Sum of HBM bin loads]\label{claim:sum-bin-loads}
The total bin load of all HBM bins is bounded by \smash{$Kc + O\left(\ell_{\max}\sqrt{K\ln N}\right)$}.
\end{claim}
\begin{proof}
Let $S_K$ denotes the sum of all HBM bin loads, and since each bin has an expected load of $c$, then $\mathbb{E}[S_K] = Kc$. We now apply Hoeffding's inequality~\cite{hoeffding1994probability} to derive a tight tail bound. As all bin loads are within $\ell_{\max}$, for any $t > 0$, we have:
\begin{equation*}
\Pr[|S_K - Kc| \geq t] \leq 2\exp\left(-\frac{2t^2}{K(\ell_{\max})^2}\right)
\end{equation*}

Setting $t = \ell_{\max}\sqrt{K\ln(2N)/2}$ leads to the aforementioned probability to be smaller than $\frac{1}{N}$.
Hence, we conclud that with probability at least \smash{$1-\frac{1}{O(N)}$}, we have \smash{$S_K \leq Kc + O\left(\ell_{\max}\sqrt{K\ln N}\right)$}.
\end{proof}

This bound is significantly tighter than the naive bound of $K\cdot\ell_{\max}$, precisely because Hoeffding's inequality captures the concentration effect when summing multiple bin loads. Next, we study a size upper bound w.r.t. the eviction stash.


\begin{claim}[stash size]\label{claim:stash-sz} Given the ratio of HBM bins as $\alpha=K/(M+K)<1$. Then with probability of at least \smash{$1-\frac{1}{O(N)}$}, the stash size does not exceed:
\smash{$\frac{(1+\alpha)M}{2} + O\left(\ell_{\max}\sqrt{M\ln N}\right)$}.
\end{claim}
\begin{proof}
We prove this by formulating the stash as a queue and analyze the queue dynamics. Let $X_t$ denote the number of elements in the queue (stash) at time $t$. There are $B=K+M$ logical bins such that $\alpha={K}/{B}<1$ as the ratio of HBM bins (those do not need evictions). At each discrete time step, we define the enqueue and dequeue strategy as follows:

\begin{enumerate}
    \item \textbf{Enqueue:} Note that an element is added to the stash only if the remapping assigns it to at least one host bin. In other words, we can formulate the enqueue strategy as an element is added to the queue with probability at most $p_{\text{add}} = 1-\alpha^2$. Moreover each added element is assigned a label uniformly at random from $\{1,2,\dots,M\}$ to record their destination.
    
    \item \textbf{Dequeue:} With probability $p_{\text{e1}} = 2\alpha(1-\alpha)$,
    a value $v\in\{1,\dots,M\}$ is chosen uniformly and every ball in the queue with label $v$ is evicted (one page read). Moreover, with probability $p_{\text{e2}} = (1-\alpha)^2$,
    two independent values $v_1,v_2\in\{1,\dots,M\}$ are chosen uniformly and every ball whose label is either $v_1$ or $v_2$ is evicted.
\end{enumerate}

 Next, we conduct drift analysis. Given that there are $x$ balls in the queue, the expected one-step change is:
\[
\begin{split}
\mathbb{E}[\Delta X_{\text{evict}} \mid X_t=x] 
&= p_{\text{e1}}\cdot\frac{x}{M} + p_{\text{e2}}\cdot\frac{2x}{M} \\
&= \frac{x}{M}\Bigl(2\alpha(1-\alpha) + 2(1-\alpha)^2\Bigr) \\
&= \frac{2(1-\alpha)x}{M}\Bigl[\alpha + (1-\alpha)\Bigr] = \frac{2(1-\alpha)x}{M}.
\end{split}
\]
Thus, the one-step drift is $\mathbb{E}[\Delta X_t \mid X_t=x] 
= (1-\alpha^2) - {2(1-\alpha)x}/{M}$. Setting the drift to zero at equilibrium, we have \smash{$x^{*} = \frac{(1+\alpha)M}{2}$}, which suggests a stable size of the eviction stash. Moreover, for any excess $\Delta > 0$, the drift becomes negative: \smash{$\mathbb{E}[\Delta X_t \mid X_t = x^{*} + \Delta] = -\frac{2(1 - \alpha)}{M}$}. This negative drift implies that once the queue exceeds \(x^{*}\), the process tends to pull it back, a self-correcting mechanism that makes larger stash sizes unlikely. In fact, as $\Delta X_t$ is bounded by the bin size (Claim~\ref{claim:bin-sz}), hence, one can apply concentration theorems~\cite{wilhelmsen1974markov, lengler2020drift, boucheron2003concentration} to derive a tail bound on $X_t - x^{*}$, which is
\smash{$O\left(\ell_{\max}\sqrt{M\ln N}\right)$}
with high probability at least \smash{\(1 - \frac{1}{O(N)}\)}. In other words, the probability that the stash size exceeds \smash{$x^{*} + O\left(\ell_{\max} \sqrt{M \ln N}\right)$} is only proportional to \smash{$\frac{1}{O(N)}$}. For completeness, we include a detailed derivation of the tail bound in Appendix~\S\ref{appx:tail}.
\end{proof}

\vspace{3pt}\noindent{\bf Dimensions in practice.} Now that we have established several analytical upper bounds on bin load, total HBM load and stash size, we aim to evaluate how tight these requirements are and determine the actual dimensions in practice. To investigate this, we conduct a validation experiment same as Ring ORAM~\cite{ren2015constants}, simulating our logical algorithm with $N=2^{20}$ data entries and subjecting it to one billion random accesses. Throughout the simulation, we track the peak load of a single bin, the aggregate load across all HBM bins, and the maximum stash occupancy. These runtime measurements are then compared against their corresponding analytical bounds. For each asymptotic term in our analytical bounds, we replace the Big-O notation with its corresponding expression multiplied by a constant factor of 1. \re{This allows us to compute concrete values that respect the stated asymptotic constraints.}
Figure~\ref{fig:validate} shows validation results under different settings (e.g., $c=8,16$ and $\alpha=0.01, 0.2, 0.5$).

\vspace{2mm}
\begin{figure}[ht]
\captionsetup[sub]{font=small,labelfont={bf,sf}}
    \begin{subfigure}[b]{0.5\linewidth}\centering\includegraphics[width=1\linewidth,interpolate=false]{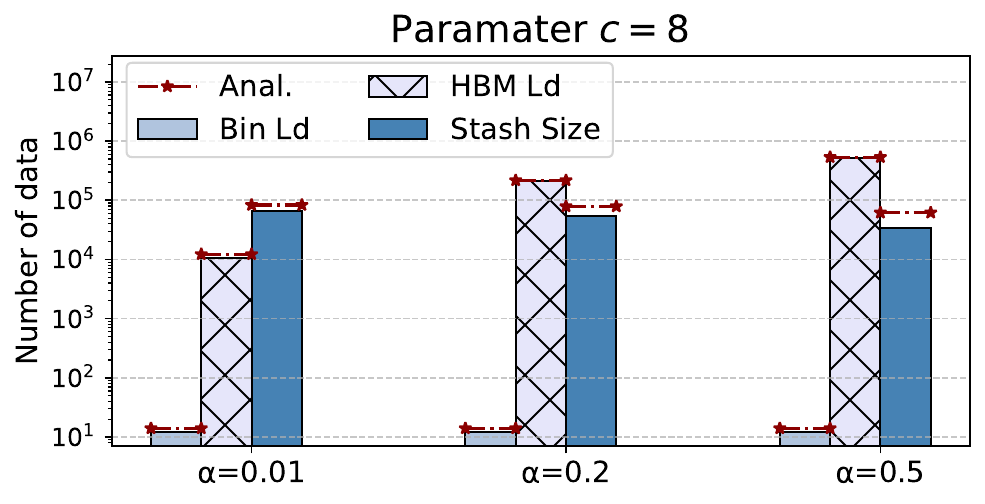}
    \end{subfigure}
    \begin{subfigure}[b]{0.5\linewidth}
\centering\includegraphics[width=1\linewidth,interpolate=false]{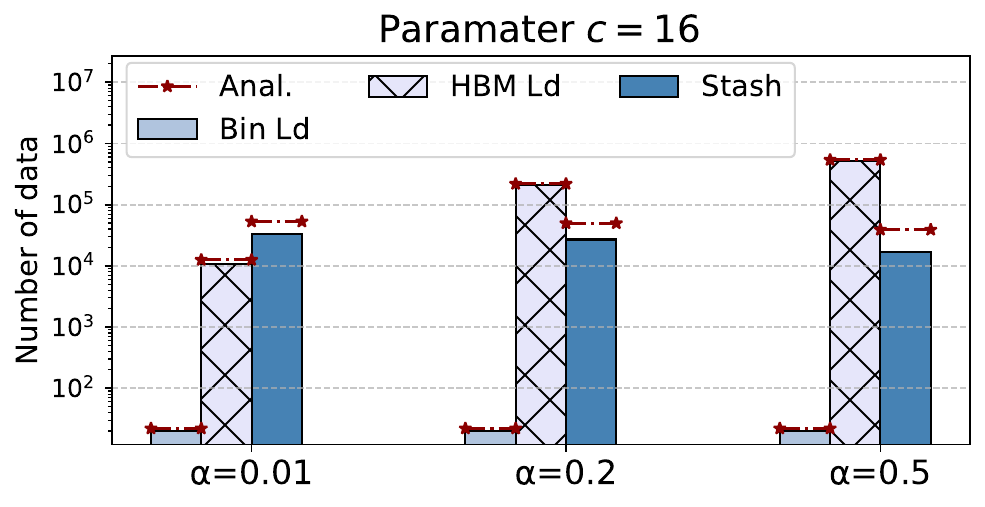}  
    \end{subfigure}
    \input{table/slack}
   \caption{Validation experiments (Exp. vs. Anal.)}
   \label{fig:validate}
\end{figure}
\vspace{1mm}

Our empirical results in Figure~\ref{fig:validate} show that the analytical bounds consistently hold across all groups, validating our upper-bound formulations. We also see that these bounds are fairly tight, with each one showing a reasonable slack compared to the corresponding empirical values. To illustrate this, we include a table in Figure~\ref{fig:validate} reporting the absolute and relative slack for each metric. Given their tightness, these bounds offer practical guidance for memory allocation to prevent overflows while retain resource-efficient.

Notably, we observe that the slack for stash sizes can be relatively large—up to 56.2\%—mainly due to our conservative analytical approach in deriving the upper bounds. Specifically, when data maps to both a host and an HBM bin, we conservatively assign it to the host stash to ensure an upper bound, though this ignores cases where it could be remapped to HBM, making the estimate pessimistic. Nevertheless, stash sizes remain small, accounting for at most 6.5\% and 3.3\% of total data for $c=8$ and $c=16$, respectively. 

\vspace{3pt}\noindent\re{{\bf Parameters selection.} Our analytical bounds help guide parameter selection for different hardware, such as varying HBM capacities. A smaller $c$ increases the position map size (\S~\ref{subsec:arch-analysis}) but reduces bandwidth overhead. Nevertheless, as the entire position map must fit in HBM, one should first choose a proper $c$ so that the map fits completely within the HBM capacity. Any remaining HBM can be used for $V_{\mathsf{hbm}}$ and the stash. A practical approach is to start with a small $\alpha$ and gradually increase it until $V_{\mathsf{hbm}}$ and the stash no longer fit in the leftover HBM. Our analytical bounds can be used to check this. Note that if $\alpha = 0$, no stash is needed.} 


%% file: table/slack.tex
\scalebox{0.63}{
\begin{tabular}{|c|lll|lll|}
\hline
\multirow{2}{*}{\textbf{\begin{tabular}[c]{@{}c@{}}Slack \\ +abs. (+rel.)\end{tabular}}} & \multicolumn{3}{c|}{\textbf{c= 8}}                                                                         & \multicolumn{3}{c|}{\textbf{c=16}}                                                                         \\ \cline{2-7} 
                                                                                          & \multicolumn{1}{c|}{\textbf{0.01}} & \multicolumn{1}{c|}{\textbf{0.2}} & \multicolumn{1}{c|}{\textbf{0.5}} & \multicolumn{1}{c|}{\textbf{0.01}} & \multicolumn{1}{c|}{\textbf{0.2}} & \multicolumn{1}{c|}{\textbf{0.5}} \\ \hline
\textbf{Bin Ld}                                                                           & \multicolumn{1}{c|}{2 (.143)}      & \multicolumn{1}{c|}{2 (.143)}     & \multicolumn{1}{c|}{2 (.143)}     & \multicolumn{1}{c|}{2 (.091)}      & \multicolumn{1}{c|}{2 (.091)}     & \multicolumn{1}{c|}{2 (.091)}     \\ \hline
\textbf{HBM Ld}                                                                           & \multicolumn{1}{l|}{1596 (.13)}    & \multicolumn{1}{l|}{7222 (.033)}  & 11630 (.022)                      & \multicolumn{1}{l|}{1877 (.15)}    & \multicolumn{1}{l|}{8419 (.038)}  & 13520 (.025)                      \\ \hline
\textbf{Stash}                                                                            & \multicolumn{1}{l|}{17511 (.21)}   & \multicolumn{1}{l|}{25555 (.324)} & 28219 (.456)                      & \multicolumn{1}{l|}{19835 (.374)}  & \multicolumn{1}{l|}{22640 (.456)} & 21895 (.562)                      \\ \hline
\end{tabular}}

%% file: tex/implementation.tex
\section{\sys Architecture}\label{sec:impl}
In this section, we introduce a concrete accelerator architecture that instantiates our logical algorithm.
\subsection{Architecture Details.}\label{sec:arch-details} To mitigate indirect leakages from the host CPU ({\bf C-2}), \sys introduces a novel {\em self-hosted isolated execution model.} While prior accelerator TEEs primarily focus on I/O isolation~\cite{volos2018graviton,zhao2022shef, oh2020trustore}, \sys goes further by migrating device control and memory management from the host (e.g., drivers) into the accelerator itself.  Concretely, \sys embeds a full-fledged OMAP logic complex behind the isolation boundary (e.g., within the chip package), which autonomously manages device I/Os, data and control flows, and access to both internal and off-package memory (e.g., host DRAM), all without relying on host-side features.  To end users, \sys exposes only a minimal instruction interface comprising two coarse-grained, task-level commands: \texttt{GET} and \texttt{PUT}. By restricting interaction to these high-level abstractions, \sys eliminates the need for fine-grained host-side coordination. As a result, the host's role is significantly reduced: it merely relays encrypted instructions and responses between the user and \sys, and provisions pinned, encrypted memory regions accessible to the accelerator.
\begin{figure}[ht]
\centering
\includegraphics[width=0.7\linewidth,interpolate=false]{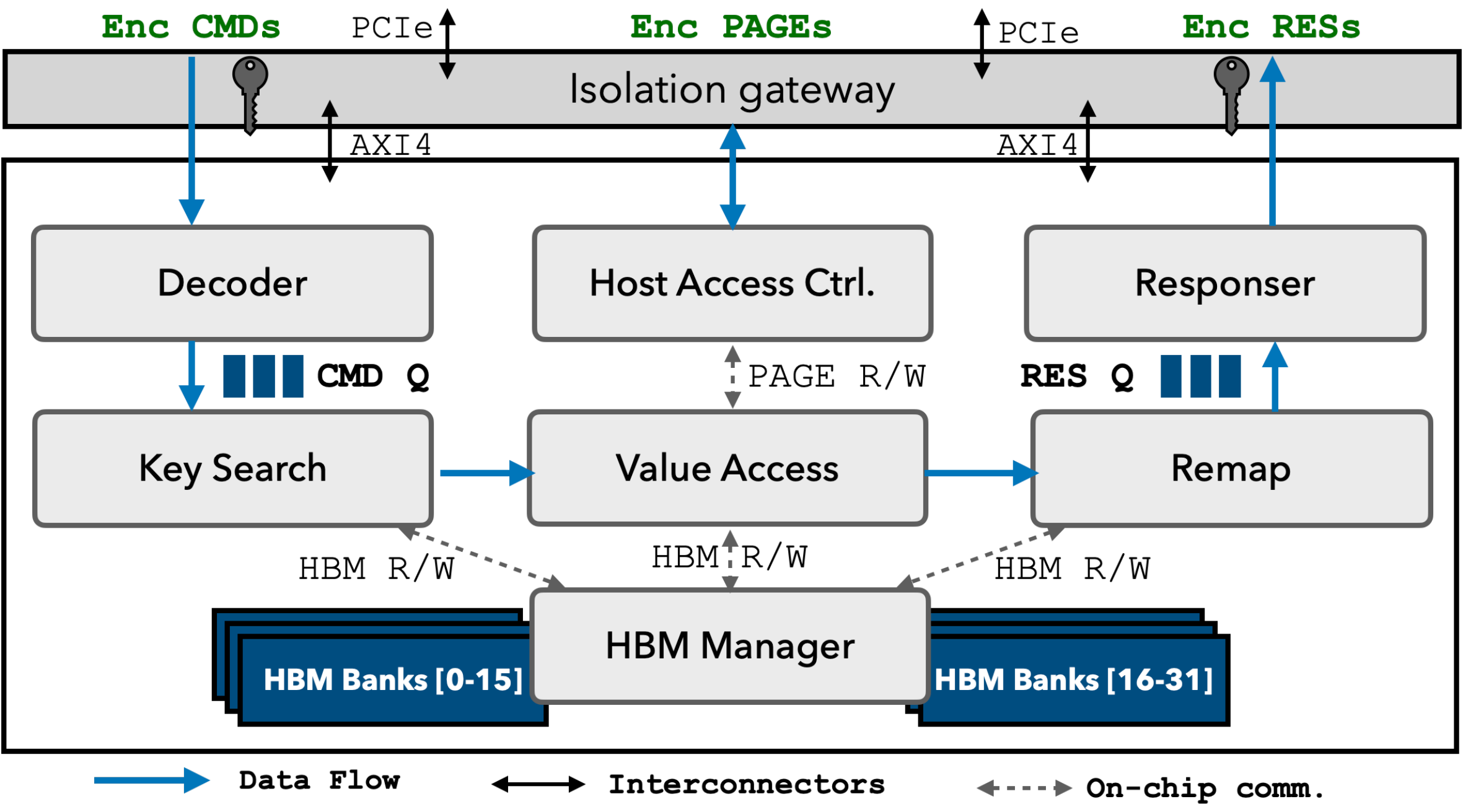}
\caption{Architectural overview of \sys.}
\label{fig:arch}
\end{figure}

Figure~\ref{fig:arch} shows an architectural overview of \sys. The execution flow consists of five main modules: decoder \texttt{(DEC)}, key search \texttt{(KS)}, value access \texttt{(VAC)}, remap (\texttt{RMP}), and responser \texttt{(RES)}. Two auxiliary modules assist for memory magagement: a host access controller (\texttt{HAC}) manages the accelerator-to-host memory accesses, and an HBM manager (\texttt{HM}) provides interfaces for other modules to access the on-package HBM banks. Next, we detail the design and execution flows of \sys. For brevity, we focus on OMAP transactions and omit standard TEE features.


\vspace{4pt}\noindent{\bf \Circled{1} Initialization.} \sys undergoes a secure boot to initialize its internal states and allocates HBM storage for position map, eviction stash, and HBM bins. The host allocates a physically contiguous, pinned memory region (e.g., Hugepages~\cite{xilinx_xrt_hm}) so that \sys can directly access it. This memory is page aligned and locked in host DRAM to prevent swapping. The base physical address of the allocated memory is then provided to the \texttt{HAC} for address translation. \sys then writes the initial pages (dummy data) via \texttt{HAC} to populate the host memory. The host also creates two I/O buffers for pooling commands and responses.


\vspace{4pt}\noindent{\bf \Circled{2} Command fetch.} \sys fetches and decrypts KV commands through the isolation gateway. Each decrypted command is a triplet of
$\langle {opcode} \mid {key} \mid {payload} \rangle$.  The 1-bit ${opcode}$ indicates whether it’s a \texttt{GET} or \texttt{PUT} operation. The payload holds the value field and is used during a \texttt{PUT} to insert, update, or delete data (deletion is triggered by a special reserved value). To prevent leakage, every instruction always includes all fields and is padded to the same fixed length. Users submit encrypted KV requests to the remote host, which relays them to \sys. The \texttt{DEC} module decodes the instruction and removes dummy payloads before passing it to \texttt{KS}.

\vspace{4pt}\noindent{\bf \Circled{3} Key search.} \sys adopts the hash based KVS design. Specifically, we maintain a global position map in HBM, implemented as a hash table, to track all inserted keys and the bins they map to. To look up a KV pair, the \texttt{KS} module hashes the input key, reads the hash entry via \texttt{HM} into the on-chip scratchpad, and searches the ``two'' mapped bins $p1$, and $p2$. Each $p_i$ denotes either an HBM address or a host page number. \texttt{KS} then forwards $p_1$, $p_2$, and the command received from \texttt{DEC} to \texttt{VAC}. If a lookup miss occurs, for instance, the key is not found in the position map, then \texttt{KS} then generates random values for $p_1$ and $p_2$, and signals to \texttt{VAC} that a new key is being inserted.

\vspace{4pt}\noindent{\bf \Circled{4} Value access.} Next, \texttt{VAC} fetches data using the bin handlers $p_1$ and $p_2$. For HBM bins, it requests \texttt{HM} to move data into an on-chip value buffer and clears the original memory line. For host-resident pages, \texttt{VAC} issues a read descriptor to \texttt{HAC}, specifying the target page number. \texttt{HAC} translates the page address, initiates a PCIe transfer to fetch the page, and stores the decrypted content into the on-chip scratchpad memory. \texttt{VAC} then scans it and places the target value into the value buffer. Note that the desired data may not be in the fetched pages and could instead reside in the eviction stash, so \texttt{VAC} also searches the stash. Once the value is retrieved, \texttt{VAC} performs the KV operation based on the command type.  For a \texttt{GET}, it writes the buffered value to \texttt{RES}. Otherwise, it updates the value buffer with the new payload (for insert/update), or clears it (for delete), and writes a confirmation code to \texttt{RES}.

\vspace{4pt}\noindent{\bf \Circled{5} Remapping and eviction.} The \texttt{RMP} module randomly selects two new bins, \( p'_1 \) and \( p'_2 \) and updates the position map with these new bins for the {\em key} that was just accessed. It then applies P2C load balancing to determine the final destination to place the data. This process is supported by an additional on-chip count list, which allows \texttt{RMP} to track the load of each bin. If the destination is an HBM bin, \texttt{RMP} issues an insertion request to \texttt{HM}, which then writes the buffered value to HBM. Otherwise, \texttt{RMP} adds the value to the eviction stash. The count list is then updated and the remapping completes. Next, \texttt{RMP} issues a page write-back (if applicable) and runs eviction. It searches the eviction stash for data mapped to the same page, adds them to the scratchpad page, and removes them from the stash. Finally, \texttt{RMP} submits a write-back descriptor to \texttt{HAC}, which initiates a PCIe transfer to overwrite the corresponding host page with the updated scratchpad page.

\vspace{4pt}\noindent{\bf \Circled{6} Response.} Once \texttt{VAC} returns the result, \texttt{RES} formats it into a fixed-length response and writes it to the host-side result buffer through the secure I/O interface. 
The response can be issued in parallel with remapping to save clock cycles.

\subsection{Co-design Optimizations}\label{sec:optimization}
Hardware designs are generally less flexible than software, which can make certain algorithms harder to implement ({\bf C-3}). To address this, \sys employs a series of co-design optimizations spanning both algorithm and hardware layers: it separates key and value storage, leverages reverse indexes for efficient eviction, integrates a specialized HBM controller for constant-cycle value operations, and optimizes memory layout to maximize HBM bandwidth. Below we discuss these in more detail.

\vspace{4pt}\noindent{\bf Decomposed HBM storage.} The host storage layout is straightforward: each fixed-size page holds multiple data tuples, each with a flag bit (to mark dummy entries), a key, and a value. The challenge lies in organizing value storage efficiently in HBM. As access patterns are hidden, maintaining a logical bin layout or dummy entries in HBM is unnecessary. Software KVSs~\cite{memcached2025,kv-application} often store keys and values together in the hash table, using linked lists to handle collisions and minimize fragmentation (Figure~\ref{fig:hash}.a). This works well because software has access to advanced abstractions like heap-allocated memory~\cite{OpenDSAHeapMemory}. Hardware, by contrast, lacks such flexibility. As a result, it typically uses fixed-size memory blocks for hash chaining~\cite{blott2013achieving}. Storing keys and values together in this context leads to significant memory waste due to large, partially unused bins (Figure~\ref{fig:hash}.b). Methods like Cuckoo hashing may reduce such overhead but require rehashing, which is hard to manage in hardware~\cite{blott2013achieving} and may leak timing information~\cite{hemenway2021alibi}.



\begin{figure}[ht]
\centering
\includegraphics[width=0.92\linewidth,interpolate=false]{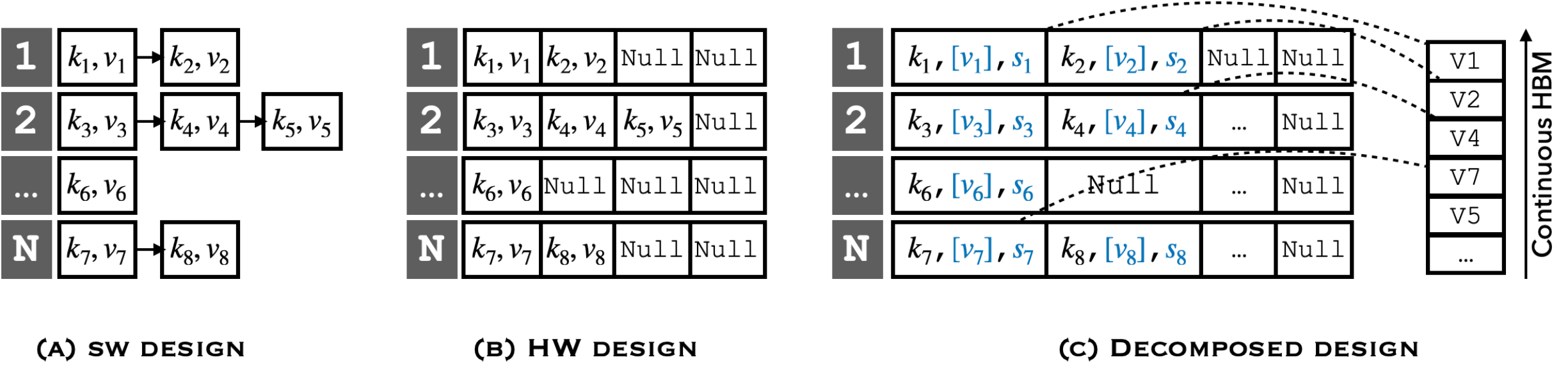}
\vspace{-1em}
\caption{Comparison of different storage design.}
\label{fig:hash}
\end{figure}

\sys resolves this challenge based on a novel decomposed storage layout. First, we store only keys in the hash table (position map) and use lighgweight indexes that reference values stored in contiguous HBM space (Figure~\ref{fig:hash}.c). Since keys are typically much smaller than values~\cite{jiang2024compassdb, apple_nsubiquitouskeyvaluestore, memcached_keysize, apify_kvstore, geeksforgeeks2024store}, the slots in the hash table remain compact, so unused entries contribute little to fragmentation in the overall KV storage. We also add a flag bit to each hash table entry to indicate whether the corresponding value resides in HBM or in the eviction stash. This way, we avoid duplicating storage for the stash. To further optimize space, we apply an aggressive load-balancing strategy to compress the position map. Specifically, we apply $d$ hash functions and assign each key to the entry with the lowest current load. According to the generalized power-of-$d$ choices theorem~\cite{mitzenmacher2001power, sitaraman2001power}, for $N$ keys and a hash table with $B$ entries, the maximum load per bin is upper bounded by \smash{$\frac{N}{B} + O\left(\frac{\log_2 \log_2 N}{\log d}\right)$}. In practice, setting $d = 4$ and $B = N/16$ suffices. In addition, all $d$ hash entries can be fetched in a burst using multiple HBM channels and searched in parallel using priority-encoder-based circuits~\cite{huang2010full, balobas2016low}, with at most an $O(\log d)$ increase in circuit depth~\cite{arora2009computational}. Hence, probing multiple entries incurs only negligible clock cycle overhead compared to searching a single entry.




\vspace{4pt}\noindent{\bf Dynamic HBM management.} While the aforementioned storage layout reduces fragmentation, it can lead to inefficient insertion costs. In the worst case, finding free space for a new value may require a linear scan of the contiguous HBM region, which leads to an $O(N)$ bandwidth overhead. To address this, we design a dynamic allocation mechanism in \texttt{HM}, using a dual-port ring buffer to efficiently track free addresses in HBM, as shown in Figure~\ref{fig:ring}.

\vspace{1mm}
\begin{figure}[ht]
\centering
\includegraphics[width=0.72\linewidth,interpolate=false]{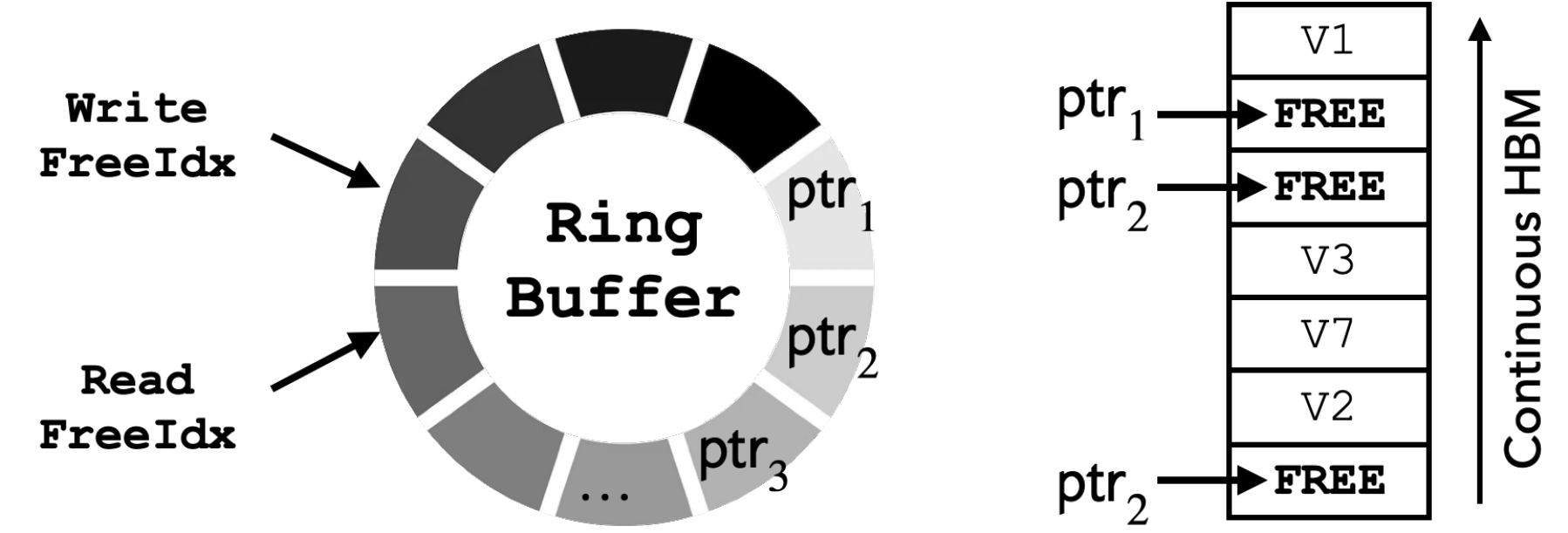}
\caption{Free address ring buffer.}
\label{fig:ring}
\end{figure}

Initially, the ring buffer is preloaded with all available HBM addresses for value stores. When space is needed for, e.g., inserting new data or remapping a page-ed value into the HBM—an address is dequeued. For other cases where data is removed—due to deletion or eviction to pages—the freed address is returned to the ring buffer. This design allows for constant-time insertions. Note that the ring buffer can remain relatively compact; for example, a 50MB buffer can address over 10 million in-HBM values. As such, the buffer can be placed on-chip rather than taking up HBM banks.


\vspace{4pt}\noindent{\bf Fast eviction with reverse index.}  A key performance bottleneck in the current design is eviction, as it requires scanning the entire position map–an \(O(N)\) operation–to locate data mapped to a specific host page. To mitigate this overhead, we introduce a lightweight reverse index: a linear table with \(M\) entries, one per host page. Each entry maintains a small list of pointers to position map entries that reference data currently staged for eviction on that page. As data placement is load-balanced, each reverse index entry holds at most \(\ell_{\max}\) pointers. Importantly, since a key’s location in the position map is stable after insertion (e.g., it is not remapped), the reverse index can be efficiently maintained. For example, during each lookup, \texttt{KS} passes the position map pointer of the accessed key to \texttt{RMP}. If the key is later added to the eviction stash, \texttt{RMP} simply adds this pointer to the corresponding reverse index entry. Similar to the ring buffer before, the reverse index is compact and can be placed on-chip to maximize performance.


\vspace{4pt}\noindent{\bf Memory optimizations.} HBM typically consists of multiple banks~\cite{AMD_Alveo_U55C} with each bank connected to its own dedicated memory channel. We leverage this architecture to enable parallel data movement and further accelerate execution. First, we allocate dedicated memory banks for commands/responses, position maps, and HBM values, thus preventing contention and enabling high-performance data movement. Additionally, for the position map table (as shwon in Figure~\ref{fig:hash}.c), we partition storage across multiple banks, with each column assigned to one bank. When \texttt{KS} loads hash entries (rows), it fetches a block from each bank simultaneously, achieving fully parallelized data movement.

\vspace{4pt}\noindent{\bf Fast data initialization.} In many cases, setting up \sys requires more than just initializing execution environments (\S~\ref{sec:arch-details}); it also requires loading outsourced KV data. An intuitive approach is to issue individual insertion requests to \sys. However, this process can be further accelerated: for instance, by letting the DO pre-process the data and organize it into $K+M$ logical bins using random mapping and P2C. A table that stores keys to bins mappings is also prepared. Both data structures are securely outsourced to the remote host. During initialization, \sys loads the outsourced data into the corresponding physical regions, sets up the position map using the mapping table, and initializes other relevant states.



\subsection{Analysis.}\label{subsec:arch-analysis}
{\bf Overhead analysis.} We analyze the overhead of \sys, focusing on two standard OMAP metrics: communication round and bandwidth blown up. 

\begin{claim}\sys incurs a constant round overhead and a total bandwidth blowup of $O(1) + O(\log_2\log_2N)$.
\end{claim}
\begin{proof} It is evident that the round overhead remains constant since the execution stages of \sys are fixed per access. Thus, we focus on the bandwidth overhead. First, the key search stage incurs an \smash{$O(\frac{N}{B}+\frac{\log_2\log_2 N}{\log_2 d})$} bandwidth overhead, as the accelerator must fetch $d$ hash entries and linearly scan them. Note that $N/B$ is a constant, and when $d\geq 4$, \smash{$\frac{\log_2\log_2 N}{\log_2 d}$} can be viewed as small as a constant. The page read, write and eviction bandwidth costs are all subject to $\ell_{\max}$, and thus is $O(c + \log_2\log_2 N)$. Remapping assigns the accessed key to new bins, requiring an update to the position map. However, the key’s entry in the map remains unchanged, allowing for an $O(1)$ update. We also insert the new value into the HBM store and update the reverse index pointer—both operations are $O(1)$, as previously discussed. Altogether, the total bandwidth overhead remains within $O(1) + O(\log_2\log_2 N)$.
\end{proof}

\vspace{4pt}\noindent{\bf HBM usage.} We now analyze the total HBM needed by \sys. Recall that three main components are stored in HBM: the position map, the HBM store, and the eviction stash. 

\begin{claim}[HBM usage]\label{claim:hbm_ratio}Let $\beta_1$ and $\beta_2$ be the upper bounds from Claims~\ref{claim:sum-bin-loads} and~\ref{claim:stash-sz}, respectively, and let $\mathit{ks}$ and $\mathit{vs}$ denote the key and value lengths in bits. Then, the total HBM usage of \sys is bounded by \smash{$(\beta_1+\beta_2)vs + (2\log_2N+ks)(N+\frac{B\log_2\log_2 N}{\log_2 d})$}.
\end{claim}
\begin{proof}
First, since HBM store and eviction stash are combined into a continuous store, the size is at most $(\beta_1+\beta_2)vs$ bits. The position map has total \smash{$(N+\frac{B\log_2\log_2 N}{\log_2 d})$} blocks, where each block contains a key, plus indexes to two logical bins. So the size is at most \smash{$(2\log_2N+ks) (N+\frac{B\log_2\log_2 N}{\log_2 d})$ bits}. Sum the two yields Claim~\ref{claim:hbm_ratio}.
\end{proof}

The HBM store is optional and used only when capacity permits; it can be disabled (e.g., $\alpha = 0$) to prioritize supporting larger datasets within limited HBM. For instance, consider a case with 1 billion data, each with a 32-bit key and a 64B value (cache-line size), and \sys is configured with $c = 8$, $B = \frac{N}{16}$, and $d = 4$. If we disable HBM store, then the total HBM required is only 26\% of the raw data size. Moreover, real-world KVSs often use small keys with large values~\cite{jiang2024compassdb, apple_nsubiquitouskeyvaluestore, memcached_keysize, apify_kvstore, geeksforgeeks2024store, kv-application}. Under such settings, the HBM usage can be further reduced—to 12\% for 256B values and 7\% for 1KB values. On the other hand, modern accelerators already offer substantial HBM capacity. For instance, HBM FPGAs like the Alveo V80 provide 32GB~\cite{amd-alveo-v80}, NVIDIA's H100 features 80GB~\cite{nvidia_confidential_computing_h100}, and newer ASICs such as AMD's MI325 offer up to 256GB~\cite{amd-mi325}. 

\vspace{4pt}\noindent{\bf Obliviousness analysis.} The main KV search logic in \sys directly follows our algorithm (Alg.~\ref{alg:omap}), and its leakage profile matches the assumptions of Alg.~\ref{alg:omap}, so that the same security guarantees as Claim~\ref{claim:obliv} holds. Initialization follows a fixed access pattern, using sequential reads/writes to load DO-prepared data into designated regions. The DO relies only on public parameters (e.g., bin sizes, upper bounds, and total data size $N$), so no data-dependent information is leaked. \sys also prevents timing leakage in KV processing: each operation (e.g., \texttt{GET}, \texttt{PUT}) executes a fixed sequence with constant-time steps. While different keys may incur different latencies (e.g., accessing HBM values vs. host memory values), this does not compromise obliviousness, as proved in Claim~\ref{claim:obliv}.

%% file: tex/evaluation.tex
\section{Evaluations}\label{sec:eva}
In this section, we detail the \sys prototype implementation and provide experiments and benchmarks to evaluate its effectiveness.

\subsection{Prototype and Testbed}\label{sec:protype}
We implement our \sys prototype on a Xilinx U55C FPGA, which features 16GB of HBM2e. The card is installed on a Dell Precision workstation in a \texttt{PCIe\_Gen3x16} slot, and the max payload size (MPS) is 512 bytes. The workstation has a 4.1GHz Intel Xeon W3-2423 CPU and 128GB of RAM. All development and experiments are conducted on this testbed, running Ubuntu 22.04.5 LTS (kernel 
 5.15.0-131-generic). We show a photo of our platform in \S~\ref{subsec:platform}.

The host runtime is implemented in C++ using the Xilinx Runtime (XRT) library (version 2.17.391) and compiled with GCC 11.4.0. It manages \sys initialization, preloads data, and handles the relay of KV commands and responses. All hardware modules are developed using High-Level Synthesis (HLS). We mainly build two kernels: \texttt{init\_kernel}, responsible for one-time initialization, and \texttt{chain\_kernel}, which streams OMAP logic to process KV commands. Memory interfaces are built using standard AXI4 busses. Data movement is handled via \texttt{m\_axi} ports, while \texttt{s\_axilite} is used for control and configuration. For host memory access, we use Xilinx Host Access Mode (HAM)~\cite{xilinx_ham}.
Internal communications, commands and responses I/Os are all implemented using \texttt{hls::stream}. The kernels are written in C++ and synthesized into \texttt{.xclbin} binaries using the default Vitis HLS flow, with \emph{no compiler optimization flags enabled.}
We target a 300 MHz clock frequency (3.33 ns period), and the design meets timing with a 3.10 ns critical path under a 0.90 ns clock uncertainty. The \sys prototype and all benchmark codes are open-sourced at: \url{https://zenodo.org/records/16905537}.

\vspace{3pt}\noindent\re{{\bf Parameters and memory settings.} Unless noted otherwise, we set $c = \frac{N}{B} = 8$ and $\alpha = \frac{K}{K+N} = 0.2$. This means each logical bin holds an average of 8 tuples, with 20\% of tuples placed in $V_{\mathsf{hbm}}$ and the rest in host memory. Both HBM and host memory are pre-allocated for each object, with sizes computed using the analytical upper bound described in \S~\ref{sec:sizes}. The synthesis process then ensures memory usage stays within these pre-allocated sizes.}

\input{table/res}
\subsection{FPGA Resource Utilization}\label{sec:area}
We report the post-route FPGA resource utilization of \sys in Table~\ref{tab:usage}. Below, we conduct detailed discussions: {\em (i) Logic resouce.} Look-Up Tables (LUTs) and Registers (REGs) are key resources used to implement control logic and manage data flow. Their combined usage typically reflects the logic complexity of a hardware design. As shown in Table~\ref{tab:usage}, the kernel-specified utilization of both LUTs and REGs remains below 6\%, indicating that \sys’s logic is simple and compact; {\em (ii) On-chip memory.} A large portion of on-chip memory remains available, with only about 11\% of BRAM utilized by \sys kernels. This memory is primarily used for on-chip buffers, indexes, and scratchpad memory during the build phase; {\em (iii) Computing resource.} \sys does not handle compute-intensive workloads, and thus it leaves all Digital Signal Processor (DSP) slices unused (the 4 slices are used by U55C shell). In general, \sys's hardware design is simple concise, and resource-efficient.
\input{table/benchmark}

\subsection{Comparison with SOTA OMAPs}\label{sec:benchmark} We benchmark \sys against two SOTA OMAPs: H2O2RAM\footnote{At the time of our experiments, H2O2RAM’s repository defaulted to an unoptimized DEBUG build. We later learned that a RELEASE build is available, which adds advanced compiler optimizations and can deliver improved performance. Nevertheless, we stress that our \sys prototype was also built without compiler optimizations. Exploring toolchain-level optimizations is beyond the scope of this paper.}~\cite{zheng2024h} and EnigMap~\cite{tinoco2023enigmap}, which represent the leading tree-based and hash-based designs, respectively. We also include a recently released industrial implementation from Facebook~\cite{facebook_oram}, which re-engineers and optimizes Oblix~\cite{oblix}. 

\vspace{4pt}\noindent{\bf Datasets and workloads.} We use a dataset containing 1 million entries, with each key being 4 bytes and each value 8 bytes~\footnote{This is the only configuration we can run EnigMap at a decent scale.}. This dataset represents the initial outsourced data loaded into the OMAPs. After init, we evaluate all systems using a YCSB-like workload~\cite{ycsb}, consisting of 2500 random \texttt{GET} and 2500 \texttt{PUT} KV operations. All commands are processed sequentially. 

\vspace{4pt}\noindent{\bf Measurements.} 
For existing OMAPs, we use their default timing interfaces to measure runtime. 
For \sys, we record elapsed time from the host side using C++’s high-resolution \texttt{clock()}, capturing both accelerator execution and PCIe round-trip latency. Since SOTA OMAPs run on CPUs with significantly higher clock frequencies (e.g., 4.1 GHz) compared to our 300 MHz \sys prototype, direct timing comparisons would be biased. Hence, we use \emph{slowdown}—the ratio of each system's performance times to that of a non-private, non-oblivious baseline KVS—as our primary metric. The baseline is written in C++ and compiled for both CPU and FPGA (with HLS-specific adjustments). The FPGA baseline uses only HBM. Due to the limited memory reporting interfaces in existing OMAP projects, we use memory usage figures from their published papers. Although all OMAPs, including ours, are designed to run with TEE support, we run experiments without them to avoid TEE-induced variability and enable a cleaner comparison of OMAP designs.


\vspace{4pt}\noindent{\bf Results.} Comparisons results are sumarized in Table~\ref{tab:e2e}. For more comprehensive comparisons, we also added two settings for \sys which captures the small ($\alpha\leq 0.01$) and large HBM (e.g., $\alpha=0.5$) cases. We begin with a complexity comparison. All SOTA OMAPs incur \( O(\log_2 N) \) rounds, and a total bandwidth overheads of \( O(\log_2^2 N) \). In contrast, \sys achieves asymptotically better complexity with constant rounds and \( O(\log_2 \log_2 N) \) overhead.

\sys's lower asymptotic overhead translates to significant efficiency gains. While the best SOTA OMAP completes the testing workload in 0.96s (5.2 KQPS), \sys finishes the same workload in just 0.023–0.032s (174–219 KQPS), achieving a raw speedup of \(30\times\) to \(480\times\). As mentioned earlier, raw latency comparison is not fair for \sys given its \(13\times\) slower clock frequency. We thus compare the slowdown measure, where \sys exhibits at most a \(2.5\times\) slowdown, while SOTA systems incur at least a \(960\times\) overhead against non-private baselines. In other word, \sys achieves a slowdown reduction of at least \(384\times\), and up to \(6338\times\) against SOTA OMAPs.


Next, we examine the init cost—the time to set up the OMAP and load initial data. Since \sys relies on the owners to pre-process data, we measure the total time of both data preparation and loading into \sys. \textsc{H2O2RAM} incurs the highest init time (291.2s) and slowdown (\(791\times\)). The reason for this stems from its need to run an extensive hash planner to determine the optimal hashing scheme~\cite{zheng2024h}. Tree-based designs initialize much faster but still suffer from slowdowns of at least \(12\times\). All \sys variants, however, exhibit near-zero slowdown, thanks to our fast init strategy where data is pre-organized and directly loaded into target regions. This results in up to a \(279\times\) speedup in raw init time and up to a \(760\times\) reduction in init slowdown against SOTA groups.


Finally, we zoom on to storage cost: \sys also reduces memory overhead, the ratio of system memory usage to raw data size, by at least \( 1.8\times \) and up to \( 10.6\times \) compared to SOTA designs. 


\subsection{Scaling Experiments} \label{sebsec:scale}

The performance of OMAPs, especially the init and query cost, is known to be sensitive to data scales~\cite{zheng2024h,tinoco2023enigmap}. To evaluate this effect on \sys, we benchmark it under varying scaling settings. 

\input{figures/scaling_fig}
\input{figures/vsz_fig}

\vspace{4pt}\noindent{\bf Experiment setup.} We adopt the same setup as \S~\ref{sec:benchmark} and consider two scaling scenarios:  {\em (i) Entry size scaling.} We fix the key and value size, but vary the number of data entries from 100K to 10M; {\em (ii) Value length scaling.} We fix the number of data entries at 1M but increase the value size from 8B to 256B, matching the largest block size evaluated by \textsc{H2O2RAM}~\cite{zheng2024h}. We focus on scaling values rather than keys, as practical systems often use small and compactly encoded keys~\cite{jiang2024compassdb, apple_nsubiquitouskeyvaluestore, memcached_keysize, apify_kvstore, geeksforgeeks2024store} but allow large values. 

\vspace{4pt}\noindent{\bf Results for entry size scaling (Figure~\ref{fig:ds}).} 
A key observation is that \sys maintains relatively stable query latency and slowdown as the number of data entries increases (Figure~\ref{fig:ds}.c,d). This stability is primarily due to \sys's \( O(\log_2 \log_2 N) \) asymptotic overhead, which grows very slowly with dataset size. In contrast, SOTA OMAPs exhibit steadily increasing query latencies as the dataset scales. As a result, \sys delivers increasingly larger performance gains on larger datasets. At 10M entries, \sys achieves at least a \(105.4\times\) speedup in raw latency and a \(1156.9\times\) reduction in normalized slowdown compared to the best SOTA design. For init cost, all SOTA OMAPs show large increases in both raw time and slowdown as data entry grows. For example, \textsc{H2O2RAM}'s init time jumps from 16s at 100K entries to over an hour at 10M, with its slowdown rising from \(51\times\) to over \(1000\times\). In contrast, \sys's raw init time increases more moderately—scaling only \(10\times\) from 100K to 10M entries. More importantly, its slowdown grows by just 9\%. As a result, for large datasets, \sys achieves substantial improvements in init efficiency, reducing slowdown by up to \(868.8\times\) compared to SOTAs.


\vspace{4pt}\noindent{\bf Results for value length scaling (Figure~\ref{fig:vsz}).}  As value size increases from 8B to 256B, all systems experience higher query times. \textsc{H2O2RAM} shows the steepest growth, becoming \(4.6\times\) slower at 256B compared to 8B. In contrast, \sys's query time increases by only \(2.1\times\) over the same range. Interestingly, the Facebook OMAP shows minimal change in query time. Nevertheless, at the 256B value size, \sys still outperforms it by \(38.7\times\) in raw query time and achieves a \(263.7\times\) reduction in slowdown. The initialization cost trends mirror those of query time: both \textsc{H2O2RAM} and \sys are more sensitive to value size changes, while the Facebook OMAP remains relatively stable. Still, \sys maintains high efficiency, requiring only 3.89s at the 256B scale, compared to 48.12s for Facebook OMAP and nearly 20 minutes for \textsc{H2O2RAM}.

\re{
\begin{table}[]
\caption{\re{Comparison with TrustOre}}
\vspace{-1em}
\label{tab:cmp-trustore}
\scalebox{0.73}{
\begin{tabular}{|c|c|c|c|c|}
\hline
                  & \textbf{Data store} & \textbf{Security} & \textbf{Throughput (QPS)} & \textbf{Latency (us)} \\ \hline
\textbf{BOLT}     & \textbf{Host+Device} & \textbf{DO}      & \textbf{209205}           & \textbf{4.8}      \\ 
\textbf{TrustOre} & On-chip Only              & Cache attacks             & 320                       & 3120.0                  \\ 
\textbf{AMD KVS}  & Host+Device          & Non-private               & 285714                    & 3.5                   \\ \hline
\end{tabular}}
\end{table}}
\vspace{-1em}
\re{
\subsection{Comparison with TrustOre}\label{subsec:trustore}
We now compare \sys with TrustOre~\cite{oh2020trustore}, an SOTA hardware ORAM controller in a heterogeneous CPU-FPGA setting. As detailed in \S~\ref{subsec:omaps}, a direct comparison between ORAMs and OMAPs is not informative. However, as TrustOre also implements map extensions~\cite{oh2020trustore}, a fair comparison is possible. We use the same benchmark settings (500 random queries with 16B key and value sizes) as TrustOre to test \sys and sample their performance figures for comparison. We also include AMD's FPGA KVS~\cite{blott2013achieving} as a non-private hardware KVS baseline. Table~\ref{tab:cmp-trustore} shows the results.
}

\vspace{2mm}\noindent\re{{\bf Results.} \sys shows significantly faster query speed than TrustOre, with over 650$\times$ improvement in query latency. This performance gap translates directly to throughput: \sys processes over 200K queries per second while TrustOre handles only 320. Notably, \sys performs even close to AMD's non-private FPGA KVS with only 36\% overhead in query latency. Beyond performance, \sys is a native OMAP with DO guarantees, while TrustOre adds map features through software algorithms dispatched on CPUs, which remain vulnerable to cache side-channels~\cite{duy2022se}. Finally, TrustOre can only store data in FPGA on-chip memory, which severely limits capacity. In contrast, \sys uses both device (on-chip and HBM) and host memory, enabling massive in-memory data support.}


\vspace{-1em}
\re{
\subsection{Micro-benchmarks}\label{subsec:micro-bench}
We analyze the cost breakdowns of \sys to find bottlenecks, especially focusing on two aspects: (i) performance, which measures each module's average running time (in clock cycles) during a single query processing; (ii) memory, which shows how much memory (in MB) is allocated to each object. In both cases, we assume a data size of 10M. As per our prior analysis (\S~\ref{subsec:arch-analysis}, \S~\ref{sebsec:scale}), varying value sizes can affect query speed and memory allocation. Hence, we report breakdowns for both the default (8B) and a larger (256B) value sizes. The results are shown in Figure~\ref{fig:t-break},~\ref{fig:m-break}.
\begin{figure}[ht]
\centering
\includegraphics[width=0.85\linewidth,interpolate=false]{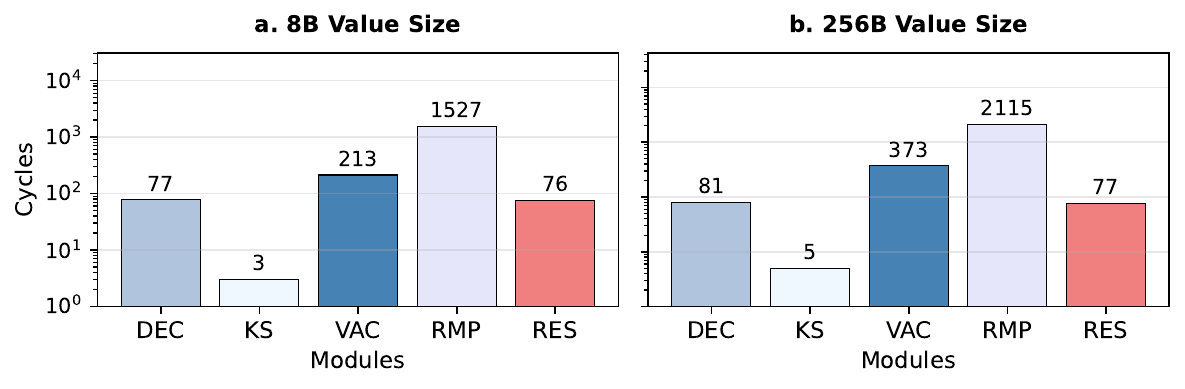}
\vspace{-1em}
\caption{\re{Performance breakdown.}}
\label{fig:t-break}
\end{figure}
\begin{figure}[ht]
\centering
\includegraphics[width=0.85\linewidth,interpolate=false]{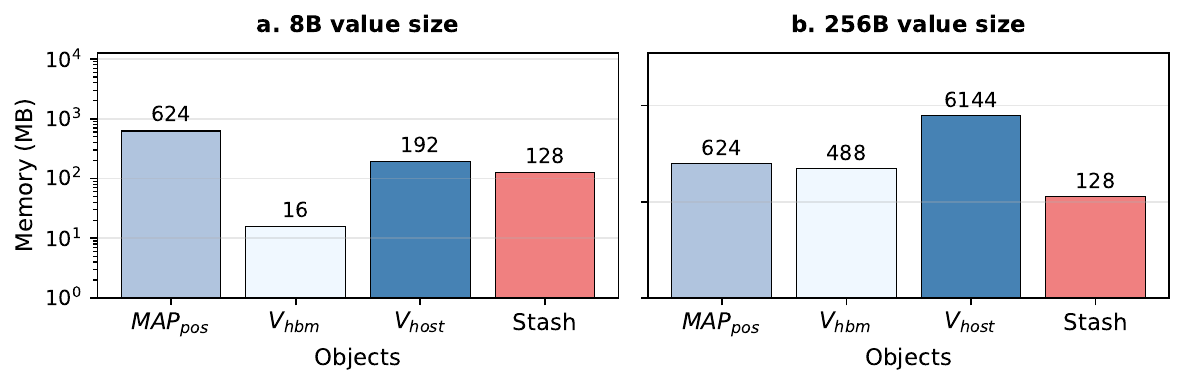}
\vspace{-1em}
\caption{\re{Memory allocation breakdown.}}
\label{fig:m-break}
\end{figure}
}

\vspace{3pt}\noindent\re{{\bf Results.} Figure~\ref{fig:t-break} shows the performance breakdowns, which reveal that the main performance bottleneck lies in the \texttt{RMP}. This is because, in \texttt{RMP}, \sys must update multiple storage objects (e.g., the position map, eviction stash, and reverse indexes) and perform stash eviction followed by the host page write-back. This observation indicates that future efforts to optimize \sys may benefit from focusing on the \texttt{RMP}. Figure~\ref{fig:m-break} shows the memory breakdown, where we can see that with 8B value sizes, the largest portion of memory is allocated to the position map. This storage cost is usually unavoidable, as the position map functions similarly to hash indexes in non-private KVSs — metadata that must be maintained to support general map operations~\cite{blott2013achieving, kv-application}. However, thanks to our decomposed memory design, both the position map and the stash do not grow with the value size, since they store only pointers to values rather than the values themselves. As a result, for larger data (e.g., 256B values), both the position map and stash account for only a small fraction of the total memory cost, with the host storage making up the majority. 
}

%% file: table/res.tex
\begin{table}[]
\caption{Max resource usage (post-route)}
\label{tab:usage}
\vspace{-2mm}
\scalebox{0.62}{
\begin{tabular}{|cl|c|c|c|c|c|c|}
\hline
\multicolumn{2}{|c|}{\textbf{Name}}          & \textbf{LUT}        & \textbf{LUTAsMem} & \textbf{REG}        & \textbf{BRAM}    & \textbf{URAM} & \textbf{DSP} \\ \hline
\multicolumn{2}{|c|}{\textbf{Total aval.}}   & 1303680             & 600201            & 2607360             & 2016             & 960           & 9024         \\ \hline
\multicolumn{2}{|c|}{\textbf{Total use}}     & 226591 {[}17.4\%{]} & 22426 {[}3.7\%{]} & 312728 {[}12.0\%{]} & 458 {[}22.7\%{]} & 0             & 4            \\ \hline
\multicolumn{2}{|c|}{\textbf{Platform}}      & 152237 {[}11.7\%{]} & 17886 {[}2.9\%{]} & 223980 {[}8.6\%{]}  & 239 {[}11.9\%{]} & 0             & 4            \\ \hline
\multicolumn{2}{|c|}{\textbf{Kernel total}}  & 74354 {[}5.7\%{]}   & 4540 {[}0.8\%{]}  & 88748 {[}3.4\%{]}   & 219 {[}10.9\%{]} & 0             & 0            \\
\multicolumn{2}{|c|}{\texttt{chain\_kernel}} & 73958 {[}5.7\%{]}   & 4540 {[}0.8\%{]}  & 88304 {[}3.4\%{]}   & 219 {[}10.9\%{]} & 0             & 0            \\
\multicolumn{2}{|c|}{\texttt{init\_kernel}}  & 396 {[}0.03\%{]}    & 0 {[}0.00\%{]}    & 444 {[}0.02\%{]}    & 0 {[}0.00\%{]}   & 0             & 0            \\ \hline
\end{tabular}}
\vspace{1mm}
\end{table}

%% file: table/benchmark.tex
\begin{table*}[]
\caption{End-to-end comparison of OMAP designs}
\vspace{-1em}
\label{tab:e2e}
\scalebox{0.85}{
\begin{tabular}{ccccccccccc}
\hline
\multicolumn{1}{|c|}{\multirow{2}{*}{\textbf{Group}}} & \multicolumn{1}{c|}{\multirow{2}{*}{\textbf{Type}}} & \multicolumn{1}{c|}{\multirow{2}{*}{\textbf{Security}}} & \multicolumn{2}{c|}{\textbf{Complexity}}                                & \multicolumn{2}{c|}{\textbf{Init time}}                                & \multicolumn{3}{c|}{\textbf{Query time}}                                                             & \multicolumn{1}{c|}{\multirow{2}{*}{\textbf{\begin{tabular}[c]{@{}c@{}}Mem\\ Overhead\end{tabular}}}} \\ \cline{4-10}
\multicolumn{1}{|c|}{}                                & \multicolumn{1}{c|}{}                               & \multicolumn{1}{c|}{}                                   & \multicolumn{1}{c|}{round}         & \multicolumn{1}{c|}{bandwidth}     & \multicolumn{1}{c|}{time (s)}      & \multicolumn{1}{c|}{slow down}    & \multicolumn{1}{c|}{time (s)}     & \multicolumn{1}{c|}{slow down}    & \multicolumn{1}{c|}{QPS (K)}      & \multicolumn{1}{c|}{}                                                                                 \\ \hline
\multicolumn{1}{|c|}{H2O2RAM}                         & \multicolumn{1}{c|}{Hash}                           & \multicolumn{1}{c|}{DO}                                 & \multicolumn{1}{c|}{$O(\log_2N)$}          & \multicolumn{1}{c|}{$O(\log^2_2N)$}          & \multicolumn{1}{c|}{291.2}           & \multicolumn{1}{c|}{791$\times$}         & \multicolumn{1}{c|}{0.96}         & \multicolumn{1}{c|}{960$\times$}       & \multicolumn{1}{c|}{5.2}            & \multicolumn{1}{c|}{12$\times$~\cite{zheng2024h}}                                                                              \\
\multicolumn{1}{|c|}{EnigMap}                         & \multicolumn{1}{c|}{Tree}                           & \multicolumn{1}{c|}{DO}                                 & \multicolumn{1}{c|}{$O(\log_2N)$}          & \multicolumn{1}{c|}{$O(\log^2_2N)$}          & \multicolumn{1}{c|}{\textcolor{gray}{4.55$^{\ddagger}$}}          & \multicolumn{1}{c|}{\textcolor{gray}{12.4$\times$$^{\ddagger}$}}        & \multicolumn{1}{c|}{\textcolor{gray}{11.41$^{\ddagger}$}}       & \multicolumn{1}{c|}{\textcolor{gray}{11410$\times$$^{\ddagger}$}}      & \multicolumn{1}{c|}{\textcolor{gray}{0.4$^{\ddagger}$}}          & \multicolumn{1}{c|}{$60\times$~\cite{zheng2024h}}                                                                            \\
\multicolumn{1}{|c|}{Facebook}                   & \multicolumn{1}{c|}{Tree}                           & \multicolumn{1}{c|}{DO}                                 & \multicolumn{1}{c|}{$O(\log_2N)$}          & \multicolumn{1}{c|}{$O(\log^2_2N)$}          & \multicolumn{1}{c|}{42.31}         & \multicolumn{1}{c|}{114.9$\times$}        & \multicolumn{1}{c|}{2.31}        & \multicolumn{1}{c|}{2310$\times$}       & \multicolumn{1}{c|}{2.1}          & \multicolumn{1}{c|}{N/A}                                                                              \\
\multicolumn{1}{|c|}{CPU baseline}             & \multicolumn{1}{c|}{Hash}                      & \multicolumn{1}{c|}{Non-private}                            & \multicolumn{1}{c|}{$O(1)$}       & \multicolumn{1}{c|}{avg. $O(1)$}       & \multicolumn{1}{c|}{0.368}         & \multicolumn{1}{c|}{--}        & \multicolumn{1}{c|}{0.001}         & \multicolumn{1}{c|}{--}        & \multicolumn{1}{c|}{4761}           & \multicolumn{1}{c|}{--}                                                                            \\\hline
\multicolumn{1}{|c|}{\textbf{BOLT (default)}}         & \multicolumn{1}{c|}{\textbf{HBM+Bin}}               & \multicolumn{1}{c|}{\bf DO}                                 & \multicolumn{1}{c|}{\textbf{$\mathbf{O(1)}$}} & \multicolumn{1}{c|}{$\mathbf{O(\log_2\log_2N)}$} & \multicolumn{1}{c|}{\textbf{1.41}} & \multicolumn{1}{c|}{\textbf{1.08$\times$}} & \multicolumn{1}{c|}{\textbf{0.028}} & \multicolumn{1}{c|}{\textbf{2.2$\times$}} & \multicolumn{1}{c|}{\textbf{174}} & \multicolumn{1}{c|}{\textbf{6.18$\times$}}                                                                                 \\
\multicolumn{1}{|c|}{\textbf{BOLT (small HBM)}}       & \multicolumn{1}{c|}{\textbf{HBM+Bin}}               & \multicolumn{1}{c|}{\bf DO}                                 & \multicolumn{1}{c|}{\textbf{$\mathbf{O(1)}$}} & \multicolumn{1}{c|}{$\mathbf{O(\log_2\log_2N)}$} & \multicolumn{1}{c|}{\textbf{1.39}} & \multicolumn{1}{c|}{\textbf{1.06$\times$}} & \multicolumn{1}{c|}{\textbf{0.032}} & \multicolumn{1}{c|}{\textbf{2.5$\times$}} & \multicolumn{1}{c|}{\textbf{155}} & \multicolumn{1}{c|}{\textbf{6.52$\times$}}                                                                                 \\
\multicolumn{1}{|c|}{\textbf{BOLT (large HBM)}}       & \multicolumn{1}{c|}{\textbf{HBM+Bin}}               & \multicolumn{1}{c|}{\bf DO}                                 & \multicolumn{1}{c|}{$\mathbf{O(1)}$} & \multicolumn{1}{c|}{$\mathbf{O(\log_2\log_2N)}$} & \multicolumn{1}{c|}{\textbf{1.35}} & \multicolumn{1}{c|}{\textbf{1.04$\times$}} & \multicolumn{1}{c|}{\textbf{0.023}} & \multicolumn{1}{c|}{\textbf{1.8$\times$}} & \multicolumn{1}{c|}{\textbf{219}} & \multicolumn{1}{c|}{\textbf{5.63$\times$}}                                                                                \\ 
\multicolumn{1}{|c|}{FPGA baseline}             & \multicolumn{1}{c|}{Hash (HBM)}                      & \multicolumn{1}{c|}{Non-private}                            & \multicolumn{1}{c|}{$O(1)$}       & \multicolumn{1}{c|}{avg. $O(1)$}       & \multicolumn{1}{c|}{1.31}         & \multicolumn{1}{c|}{--}        & \multicolumn{1}{c|}{0.013}         & \multicolumn{1}{c|}{--}        & \multicolumn{1}{c|}{381}           & \multicolumn{1}{c|}{--}                                                                            \\\hline

\multicolumn{11}{l}{\footnotesize $\ddagger$. We were unable to run EnigMap for full data size, likely due to our memory capacity limitations, so we report its results at $N=260K$, which is the largest possible size we can complete.} 
\end{tabular}}
\vspace{-1em}
\end{table*}

%% file: figures/scaling_fig.tex
\eat{\begin{figure*}[h]
\captionsetup[sub]{font=small,labelfont={bf,sf}}
    \begin{subfigure}[b]{0.248\linewidth}
    \centering    \includegraphics[width=1\linewidth]{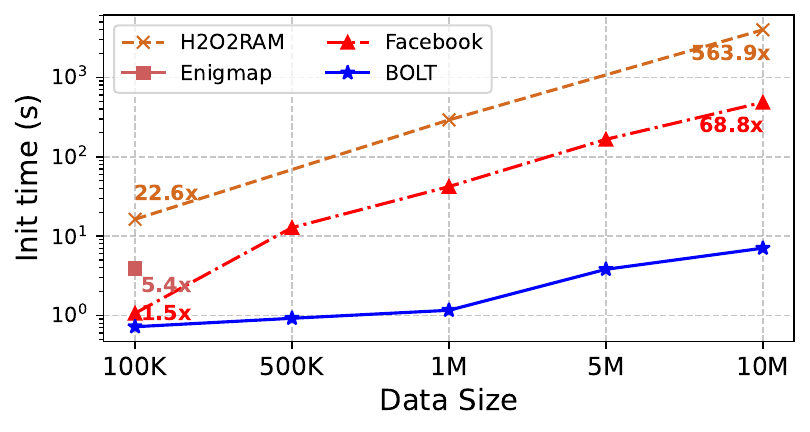}
    \caption{Init time}
    \label{fig:ds-init}
    \end{subfigure}
    \begin{subfigure}[b]{0.248\linewidth}
    \centering    \includegraphics[width=1\linewidth]{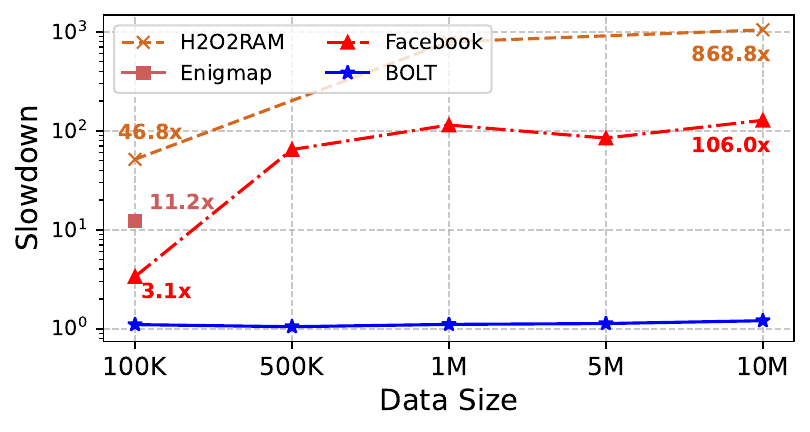}  
    \caption{Init slowdown}
    \label{fig:ds-init-sd}
    \end{subfigure}
    \begin{subfigure}[b]{0.248\linewidth}
    \centering    \includegraphics[width=1\linewidth]{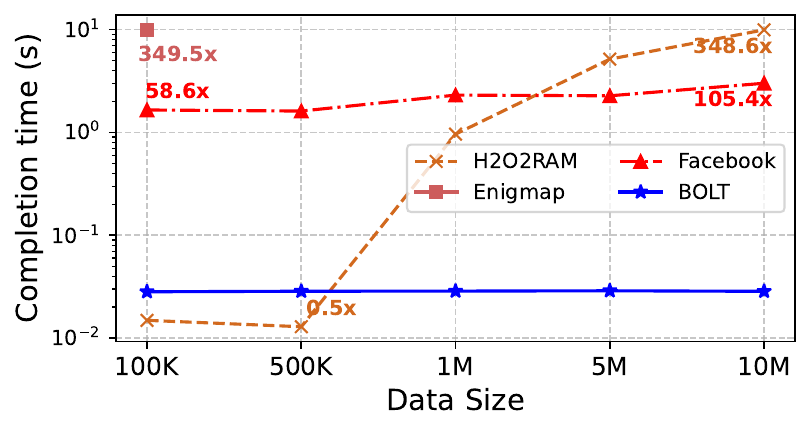}
    \caption{Query time}
    \label{fig:ds-q}
    \end{subfigure}
    \begin{subfigure}[b]{0.248\linewidth}
    \centering    \includegraphics[width=1\linewidth]{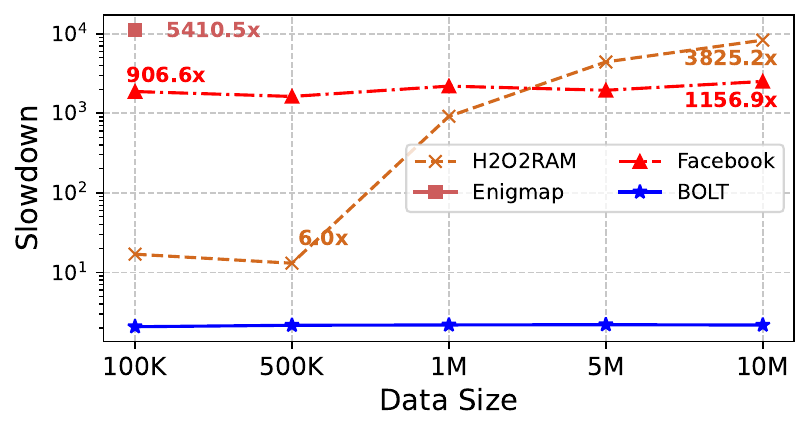}  
    \caption{Query slowdown}
    \label{fig:ds-q-sd}
    \end{subfigure}
    \vspace{-2mm}
   \caption{Scaling Experiment Scenario 1: OMAP performance under increasing data entries.} \vspace{-2mm}
   \label{fig:ds}
\end{figure*}}

\begin{figure}[ht]
\centering
\includegraphics[width=\linewidth,interpolate=false]{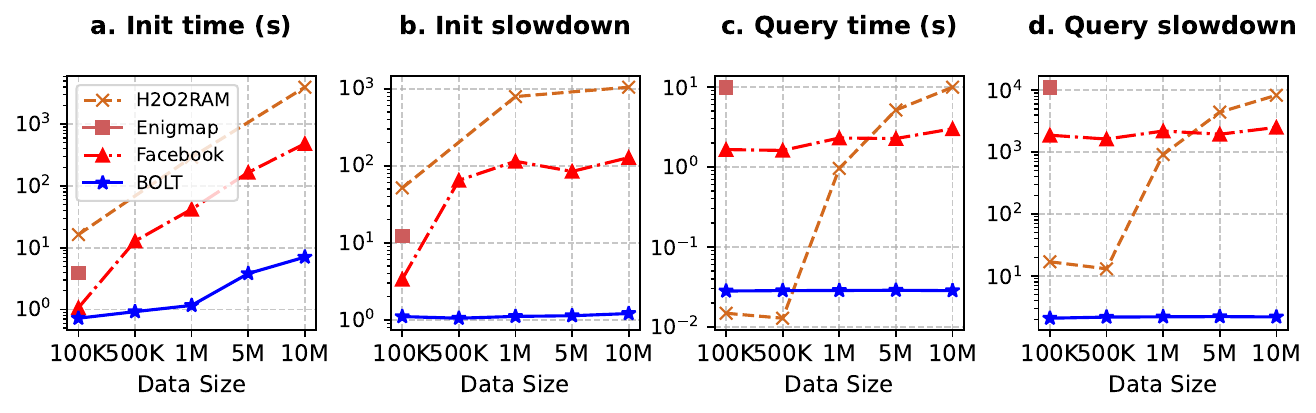}
\vspace{-2.2em}
\caption{Performance under scaling data entries.}
\label{fig:ds}
\end{figure}

%% file: figures/vsz_fig.tex
\eat{\begin{figure*}[h]
\captionsetup[sub]{font=small,labelfont={bf,sf}}
    \begin{subfigure}[b]{0.248\linewidth}
    \centering    \includegraphics[width=1\linewidth]{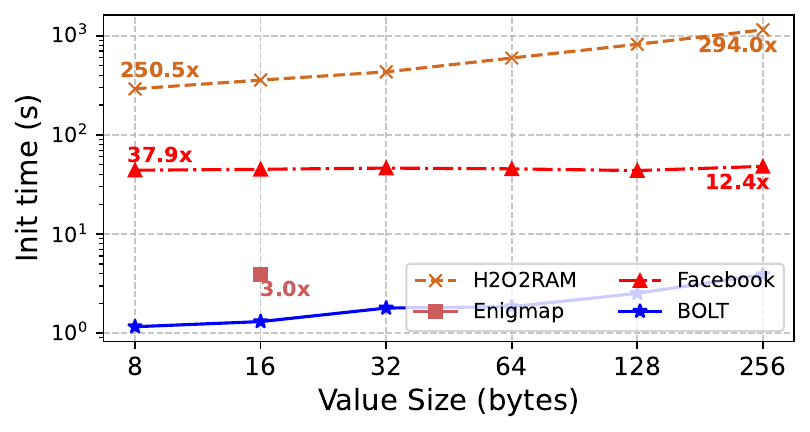}
    \caption{Init time}
    \label{fig:vsz-init}
    \end{subfigure}
    \begin{subfigure}[b]{0.248\linewidth}
    \centering    \includegraphics[width=1\linewidth]{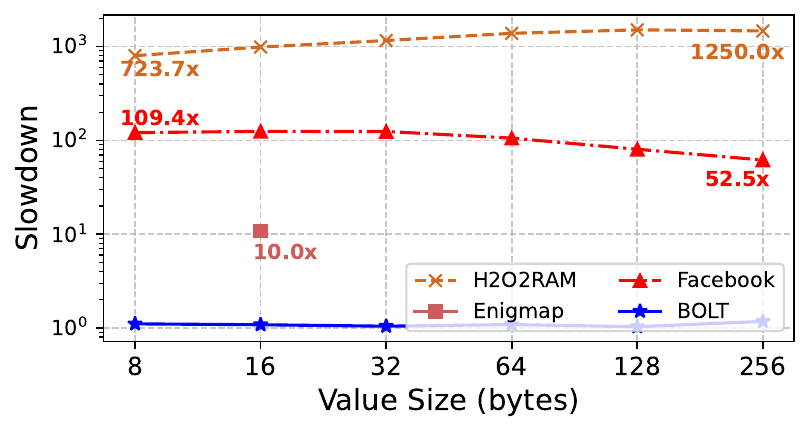}  
    \caption{Init slowdown}
    \label{fig:vsz-init-sd}
    \end{subfigure}
    \begin{subfigure}[b]{0.248\linewidth}
    \centering    \includegraphics[width=1\linewidth]{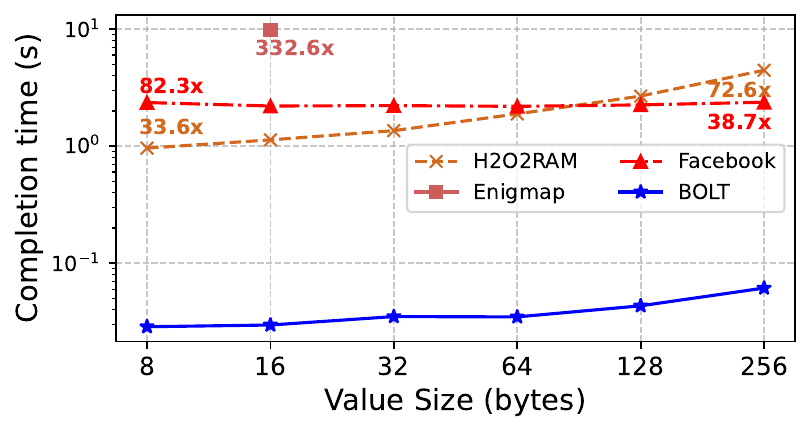}
    \caption{Query time}
    \label{fig:vsz-q}
    \end{subfigure}
    \begin{subfigure}[b]{0.248\linewidth}
    \centering    \includegraphics[width=1\linewidth]{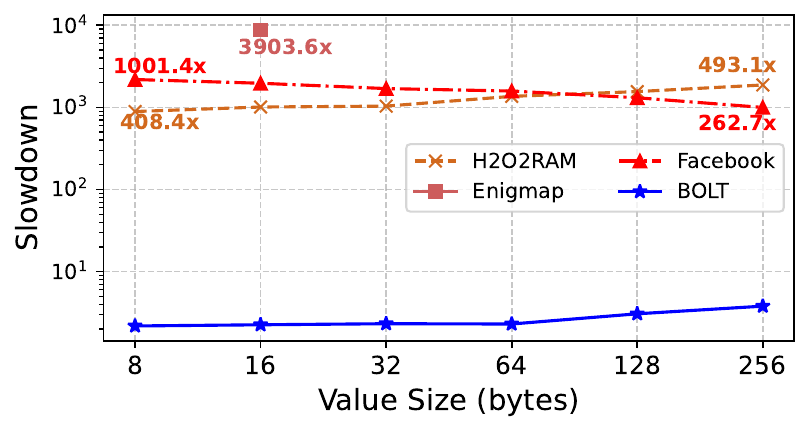}  
    \caption{Query slowdown}
    \label{fig:vsz-q-sd}
    \end{subfigure}
    \vspace{-2mm}
   \caption{Scaling Experiment Scenario 2: OMAP performance under increasing value sizes.} \vspace{-4mm}
   \label{fig:vsz}
\end{figure*}
}

\begin{figure}[ht]
\centering
\includegraphics[width=\linewidth,interpolate=false]{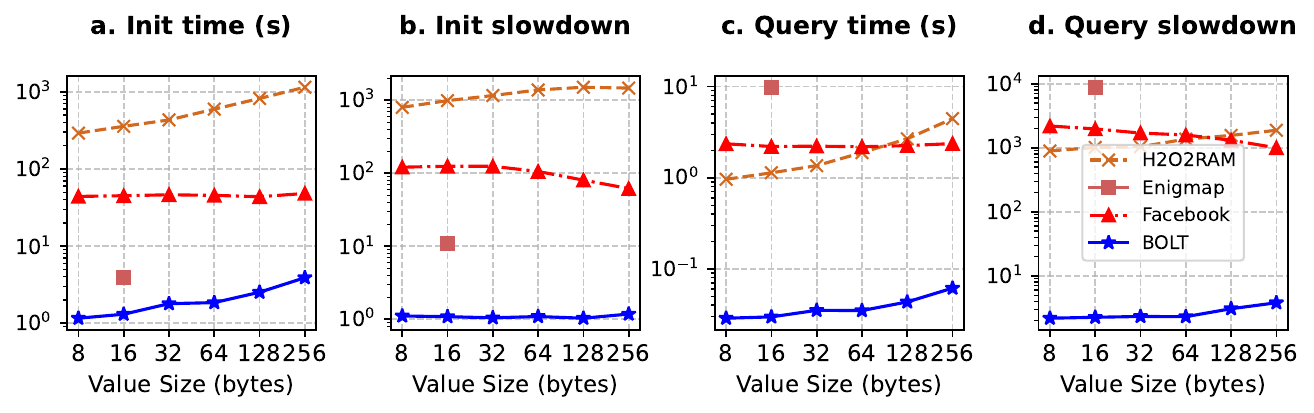}
\vspace{-2.2em}
\caption{Performance under scaling value sizes.}
\label{fig:vsz}
\vspace{-.2em}
\end{figure}

%% file: tex/related.tex
\section{Related Work}
\label{sec:related}
\noindent{\bf ORAMs and OMAPs.} A survey of ORAMs and OMAPs is provided in \S~\ref{subsec:omaps}; here, we focus on distinguishing \sys from existing designs. Since the seminal work on ORAMs~\cite{goldreich1987towards,goldreich1996software}, it is well established that any ORAM must incur an amortized bandwidth blowup of at least $\Omega(\log_2 N)$~\cite{path-oram, ren2015constants, resizable-tree-based-oram, bindschaedler2015practicing, sasy2017zerotrace, fletcher2015low, asharov2023futorama, asharov2020optorama, patel2018panorama, dittmer2020oblivious, stefanov2013oblivistore}. This lower bound heavily impacts later OMAPs, which typically build on ORAM primitives, leading to $O(\log_2 N)$ rounds and $O(\log_2^2 N)$ bandwidth overhead in SOTA designs~\cite{oblix,chamani2023graphos,tinoco2023enigmap,facebook_oram, zheng2024h}. Nevertheless, the $\Omega(\log_2 N)$ result is derived under the classical RAM model, which assumes that only the CPU registers are physically shielded, while all other components are subject to access pattern leakages~\cite{goldreich1987towards}. As a result, it is naturally assumed that the available unobservable memory is constant in size, limited to a fixed number of CPU registers. This assumption, however, breaks down on modern accelerator architectures, which often feature large memory dies~\cite{amd-alveo-v80, amd-mi325, nvidia_confidential_computing_h100} stacked within the chip package. These on-chip memories share similar physical properties with registers and, with proper isolation~\cite{volos2018graviton, hunt2020telekine}, can be rigorously shielded to serve as unobservable memory. \sys takes advantage of this architectural shift by using large unobservable HBM to design new OMAP algorithms that go beyond classical bounds, achieving constant rounds and $O(1) + O(\log_2 \log_2 N)$ bandwidth overhead.



\vspace{4pt}\noindent\re{{\bf Secure memory hardware.} Several works have explored secure memory hardware, generally taking one of two main approaches. The first approach focuses on accelerating ORAMs with FPGAs or ASICs by implementing existing algorithms as bare-metal secure memory controllers~\cite{fletcher2015low, liu2015ghostrider, maas2013phantom, ye2025palermo, che2020multi}. However, these designs remain subject to the inherent $\Omega(\log_2 N)$ bandwidth lower bound. The second approach uses specialized memory cubes~\cite{awad2017obfusmem, aga2017invisimem, oh2020trustore, duy2022se, choi2024shieldcxl} to build unobservable memory. While this avoids the $\Omega(\log_2 N)$ overhead, it faces key limitations: constrained memory capacity ({\bf C-1}) and potential indirect leakage through the host ({\bf C-2}). \sys addresses all these limitations. Moreover, prior efforts focus solely on secure memory extensions for address-value pair accesses, whereas \sys is a native OMAP accelerator specifically designed for KVS.}

\vspace{4pt}\noindent{\bf Accelerator TEEs.} Recent research~\cite{zhao2022shef,armanuzzaman2022byotee,volos2018graviton,hunt2020telekine, mai2023honeycomb, wang2024towards,jang2019heterogeneous, ivanov2023sage} and industry products~\cite{nvidia_confidential_computing_h100, vaswani2022confidential, dhar2024ascend} have driven growing interest in accelerator TEEs. However, their focus is primarily on ensuring isolation and integrity, rather than rigorous data-obliviousness. Hunt et al.~\cite{hunt2020telekine} highlight that while isolated HBM improves security, it does not guarantee obliviousness because indirect leakage from the host remains possible. Their solution offloads control functions to a trusted client, but this introduces significant communication overhead. In \sys, we take a fundamentally different self-hosted approach that achieves the same goal as~\cite{hunt2020telekine} but without relying on a trusted client. Moreover, prior accelerator TEEs, including Hunt et al., have mainly focused on compute-intensive ML and scientific workloads, which tend to have well-structured access patterns. In contrast, \sys targets memory-intensive KVS workloads.


%% file: tex/conclusion.tex
\vspace{-1em}
\section{Conclusion}
\label{sec:conclusion}
In this work, we take the first step toward leveraging architectural advancements in modern accelerators to design OMAPs that are both secure and efficient. Specifically, the emergence of HBM in accelerators allows us to build large HUMs, breaking the long-standing assumption in oblivious primitive designs that such memory regions must be constant-sized. By exploiting this shift, our prototype \sys achieves strong performance—up to 352$\times$ faster than SOTA OMAPs—while maintaining practicality, with overheads as low as $1.7\times$ compared to non-private KVSs.

%% file: tex/acknowledgements.tex
\section*{Acknowledgements}
We extend our sincere gratitude to our shepherd and the anonymous reviewers for their invaluable feedback and constructive suggestions. We also wish to thank the members of CDCC, as well as Intel Trustworthy Data Center of the Future for their generous support. This work was supported in part by the National Science Foundation under awards OAC-2419821 and CNS-2207231, the Intel Trustworthy Data Center of the Future grant, and the AMD University Program for providing us with the U55C FPGA card. Any opinions, findings, and conclusions or recommendations expressed in this material are those of the author(s) and do not necessarily reflect the views of the National Science Foundation, Intel, or AMD.

%% file: tex/appendix.tex
\appendix
\newcommand\hongbo[1]{\textcolor{teal}{\{\textbf{hongbo:} {\em#1}\}}}

\section{Evaluation continued}\label{sec:eval}
\subsection{Testbed and Prototype}\label{subsec:platform}
We provide additional information about our testbed and prototype. Specifically, Figure~\ref{fig:platform} shows a photo of our testbed platform, followed by a prototype gate-level schematic in Figure~\ref{fig:gate} that illustrates the post-synthesis netlist, including logic gates, flip-flops, and other hardware primitives.

\vspace{1mm}
 \begin{figure}[ht]
\centering
\includegraphics[width=0.63\linewidth,interpolate=false]{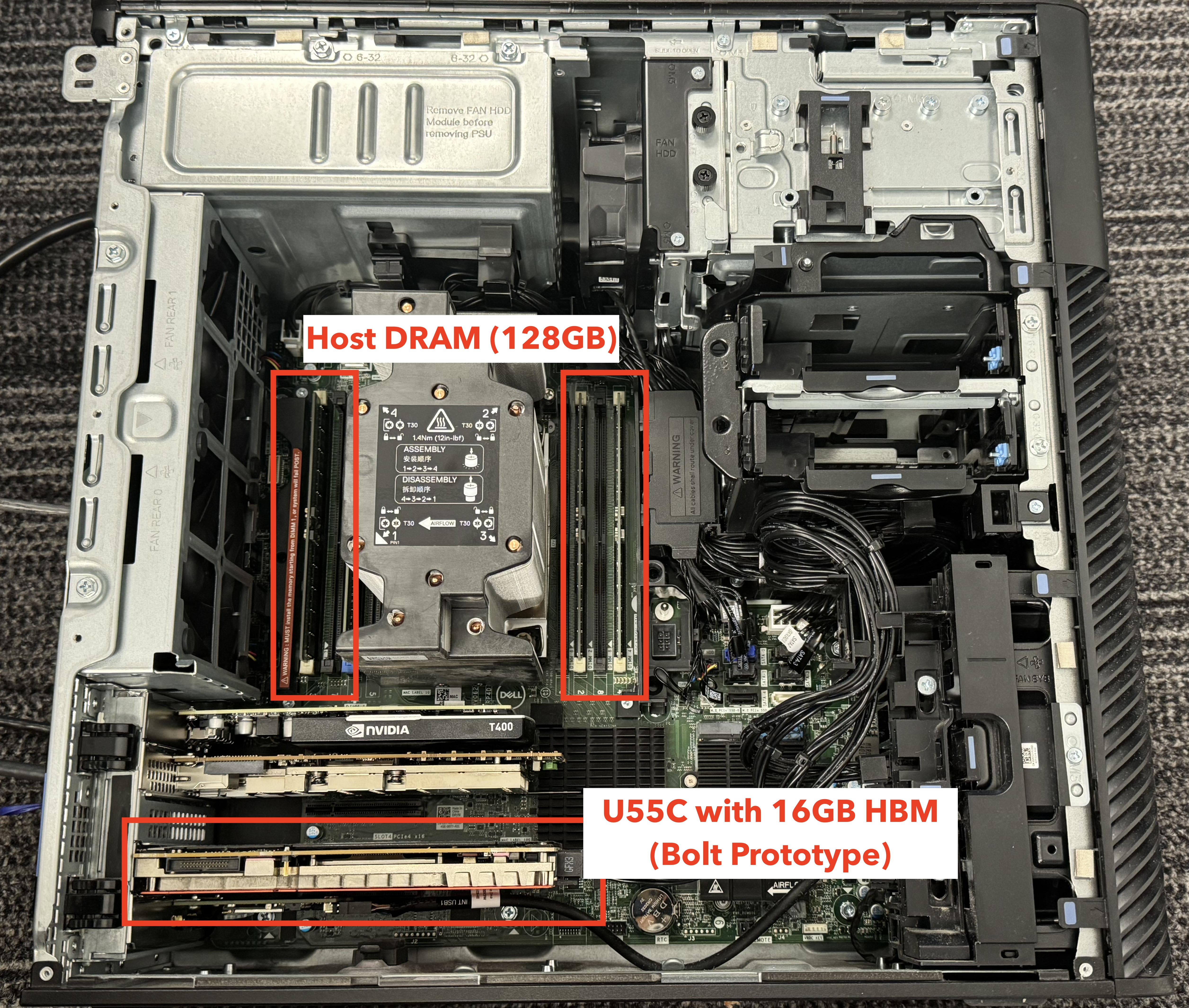}
\caption{The testbed and \sys prototype.}
\label{fig:platform}
\end{figure}

\vspace{1mm}
 \begin{figure}[ht]
\centering
\includegraphics[width=0.63\linewidth,interpolate=false]{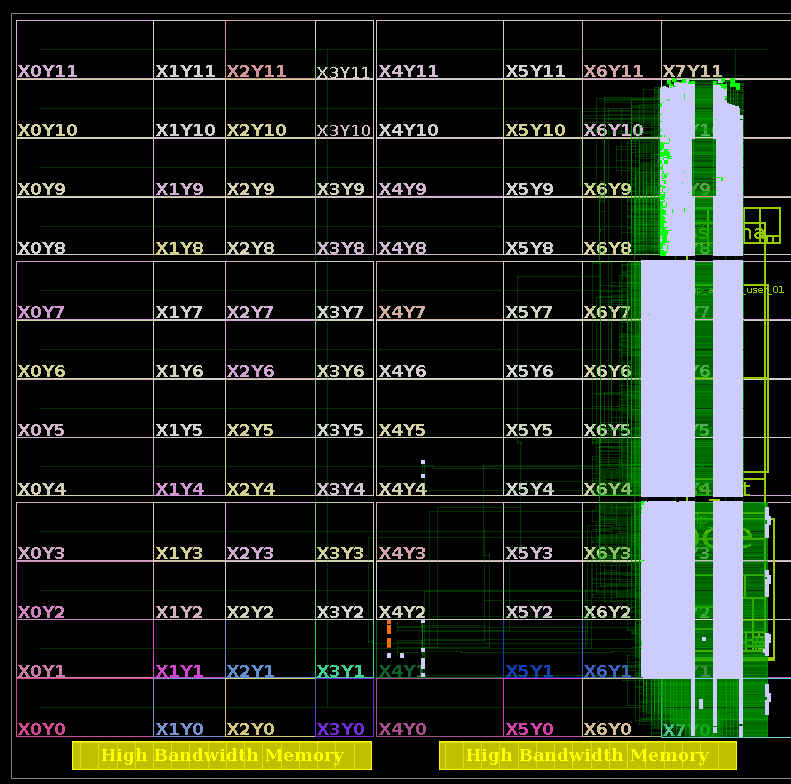}
\caption{Post-synthesis netlist schematic of \sys.}
\label{fig:gate}
\end{figure}

\section{Additional Background }
\subsection{FPGA and its security features}
\vspace{3pt}\noindent{\bf FPGA.}  
An FPGA is a hardware device consisting of configurable logic blocks and interconnects, programmable by loading a developer-created binary file called a bitstream.
The bitstream, created by specialized FPGA design software, describes the exact logical operations and connections required to realize custom micro-architecture design. Current FPGA manufacturers already introduce important security features including (1) Hardware root of trust (HWRoT) (2) Bitstream encryption and (3) Secure boot.



\vspace{3pt}\noindent{\bf HWRoT.}  A  HWRoT is a compact, tamper-resistant hardware module embedded in silicon that serves as the foundation for a system’s security functions. 
It comprises two primary components:
(1) {\em Boot ROM.} An immutable section of code that executes immediately upon power-up to establish the initial chain of trust.
(2) {\em Cryptographic elements.} These include unique device identifiers, signing keys, and root derivation keys, which are securely stored in isolated hardware structures such as one-time programmable eFUSEs, embedded key ROMs, or battery-backed secure RAM.

These storage mechanisms are designed to prevent software access, resist physical tampering, and enforce immutability, thereby ensuring that sensitive cryptographic materials remain secure throughout the device’s lifecycle. In programmable accelerators such as FPGAs and GPUs, the HWRoT handles critical operations including secure boot, bitstream decryption, runtime authentication, and secure configuration.

\vspace{3pt}\noindent{\bf Bitstream encryption.} Specifically, bit-stream encryption protects bitstreams during transmission and storage using AES encryption, preventing unauthorized disclosure, copying, or reverse-engineering~\cite{zhang2019comprehensive,yoon2018bitstream,danesh2020turning}. During manufacturing, FPGA vendors securely generate cryptographic keys, including an AES encryption key ($\mathsf{K}_{bit}$) and  RSA key pairs consisting of a private key ($\mathsf{sk}_{bit}$) and a public key ($\mathsf{pk}_{bit}$).  The AES key and RSA public key are embedded into one-time programmable, non-volatile storage known as eFUSEs inside the FPGA hardware. Prior to deployment, the bitstream is encrypted with AES encryption key ($\mathsf{K}_{bit}$) and digitally signed with the  RSA private key ($\mathsf{sk}_{bit}$), ensuring both confidentiality and authenticity.

\vspace{3pt}\noindent{\bf Secure boot.} 
In FPGA accelerators, secure boot has been proposed as a mechanism for runtime integrity~\cite{oh2020trustore}, as it ensures that the loaded bitstream is authentic and that the entire micro-architecture is correctly configured. At boot time, the FPGA loads the encrypted and signed bitstream from external storage, authenticates it using the embedded RSA public key ($\mathsf{pk}_{bit}$), and upon successful verification, decrypts the bitstream with the AES encryption key ($\mathsf{K}_{bit}$). After decryption, the FPGA's bootloader securely loads this verified bitstream into the reconfigurable hardware that will actually run the intended functions. Critical components involved in secure boot—including the boot ROM, cryptographic keys, and AES decryptor—are designed to be tamper-resistant, relying on secure hardware provided by FPGA manufacturers.

\subsection{Accelerator TEEs}\label{appx:acc-tee}

TEEs are secure execution environments isolated from the normal operational environment to protect sensitive code and data from unauthorized access or tampering. A comprehensive TEE provides isolation, confidentiality, integrity, and remote attestation guarantees. Recently, TEEs have been extended beyond traditional CPUs to include accelerators such as FPGAs and GPUs. Accelerator TEEs are designed to offload and safeguard compute or memory-intensive workloads that require runtime confidentiality and integrity.

\vspace{3pt}\noindent\textbf{Remote attestation (RA).} RA enables external entities to verify that an accelerator TEE is correctly configured and executing trusted code~\cite{sgx,zhao2022shef,oh2020trustore, nvidia_confidential_computing_h100}. 

 The generalized RA workflow for accelerator TEEs includes the following steps: (1) Key provisioning: A pair of attestation keys is prepared in advance. These keys may
be directly fused into the device by the manufacturer or derived from HWRoT. In other words, we consider these
keys to be non-forgable by malicious attackers. The private key ($\mathsf{sk}_{att}$) 
is securely stored inside the accelerator, typically in secure storage like
eFUSEs or boot ROM, while the public key ($\mathsf{pk}_{att}$)  is held and managed
by the manufacturer; (ii) Challenge and response: When attestation is
initiated, the user sends a randomly generated challenge to the accelerator.  The accelerator then signs a measurement report, which includes the
challenge, a snapshot of its runtime state (e.g., loaded firmware hash),
and a unique device ID, using $\mathsf{sk}_{att}$. This signed report is returned to the
user; (iii) Verification: The user forwards the signed report to the manufacturer  via a secure channel. The manufacturer verifies the signature
using the corresponding public key and checks whether the reported
runtime state matches an expected trusted configuration.



\vspace{3pt}\noindent\textbf{ Confidentiality.} 
I/O isolation is a prevalent method for establishing TEEs on modern accelerators~\cite{zhao2022shef,oh2020trustore,nvidia_confidential_computing_h100,volos2018graviton,hunt2020telekine,armanuzzaman2022byotee}. This approach implements a hardware firewall to restrict direct external access, channeling all device I/O operations—such as Memory-Mapped I/O (MMIO), Direct Memory Access (DMA), and AXI interfaces—through secure interfaces. Within this isolated environment, confidential data and code are decrypted and processed exclusively inside the hardware firewall, ensuring sensitive information remains protected. Data exiting this isolated space is re-encrypted to maintain confidentiality during transit or storage.

In FPGA designs, the isolation firewall is typically provided by the manufacturer or a trusted vendor as a customized shell extension that sits behind the standard manufacturer shell. The initialization workflow is as follows~\cite{zhao2022shef}:
(1) Key Provision: The FPGA manufacter
generates a public/private asymmetric encryption key pair ($\mathsf{sk}_{\mathsf{shield}}$,$\mathsf{pk}_{\mathsf{shield}}$) before deployment. The private encryption key is embedded into the firewall bitstream, then the bitstream is encrypted as mentioned in the bitstream encryption section.  At run‑time, the key is then directly loaded to on-chip registers through the FPGA secure boot and thus is considered confidential.  The public key 
is shared with the data owner through secure and authenticated channels.
(2)Secure Data Encryption Keys Provision: The data owner generates one or more symmetric keys ($\mathsf{K}{sec}$). These keys are used to secure all communications with the remote TEE, encrypting confidential data stored outside the TEE and decrypting data inside the TEE.
Each DEK is encrypted under 
$\mathsf{pk}_{\mathsf{shield}}$, yielding a ``load key'' blob.
Once the enclave has completed secure boot and remote attestation, the host transmits the blob to the enclave over the authenticated channel. 
(3) Runtime  Encryption and Decryption: 
Once the key blob is in place, the hardware firewall decrypts  $\mathsf{K}_{sec}$ and configures it to transparently encrypt and decrypt all I/O operations. Specifically, the firewall exposes the same interfaces as traditional I/O mechanisms, such as MMIO or DMA, but proxies the traffic, for instance, decrypting inbound messages and encrypting outbound data.

For GPU and ASIC TEEs~\cite{nvidia_confidential_computing_h100, dhar2024ascend}, the overall design concepts are similar to that of FPGA-based TEEs. The primary difference is that these devices typically derive I/O encryption keys internally from their HWRoT. Additionally, some designs rely on software firewalls, rather than hardware-based solutions, to establish the isolated region~\cite{mai2023honeycomb}. 

\vspace{3pt}\noindent\textbf{ Encryption integrity.} Beyond ensuring confidentiality, TEEs must also guarantee integrity—particularly to ensure that data encrypted and sealed outside the enclave has not been tampered with by malicious users.

Common integrity protection techniques include authenticated encryption schemes like AES-GCM~\cite{gueron2017aes}, which combine encryption with a Message Authentication Code (MAC) to detect tampering~\cite{zhao2022shef,nvidia_confidential_computing_h100}. For large data regions, data is split into fixed-size chunks (e.g., 4KB), each independently encrypted and authenticated to prevent block reordering or substitution. To defend against replay attacks, each chunk's MAC is computed using a monotonic counter, and a lightweight Merkle tree is constructed over these MACs and counters~\cite{gueron2016memory}. The tree’s root hash, securely stored on-chip, commits to the state of the entire data.

\section{Proof of Theorems}
\subsection{Proof of Claim~\ref{claim:bin-sz}}\label{appx:p2c}
To prove this, we consider the classical balls‐and‐bins model for allocating \(N\) balls into \(B=K+M\) bins and use use a layered induction argument to prove the claim. For each integer $i$, define $X_i = \#\{ \text{bins with load} \ge i \}$,
after all $N$ balls are placed. When a ball is placed, it selects two bins uniformly at random and is placed into the less loaded one.

For a ball to increase a bin's load from $i$ to $i+1$, \emph{both} selected bins must have load at least $i$. Thus, if at some stage there are $X_i$ bins with load at least $i$, then the probability that a given ball increases some bin’s load from $i$ to $i+1$ is at most
\smash{$\left({X_i}/{B}\right)^2$}. Since there are $N$ balls, by linearity of expectation we have
\smash{$\mathbb{E}[X_{i+1}] \le N \left({X_i}/{B}\right)^2$}.

We claim that for all integers $k \ge 0$, with probability at least 
$1 - \frac{1}{N2^k}$, the number of bins with load at least $c+k$ satisfies $X_{c+k} \le \beta_k$,
where the threshold sequence $\{\beta_k\}$ is defined by
\[
\beta_0 = B \quad \text{and} \quad \beta_{k+1} = 2N\left(\frac{\beta_k}{B}\right)^2.
\]

\subsubsection*{Base Case ($k=0$)}
For $i=c$, it is trivial that $X_c \le B = \beta_0$,
as every bin is counted and the average load is $c$.

\subsubsection*{Inductive Step} Assume that with probability at least $1-\frac{1}{N2^k}$ it holds that $X_{c+k} \le \beta_k$.
Then for a given ball, the probability that both choices lie in the set of $\beta_k$ bins is at most $(\beta_k/B)^2$. And, over $N$ balls, $\mathbb{E}[X_{c+k+1}] \le N\left({\beta_k}/{B}\right)^2$. Applying the multiplicative Chernoff bound with \(\delta=1\) yields
\[
\Pr\left[X_{c+k+1} \ge 2N\Bigl(\frac{\beta_k}{B}\Bigr)^2\right] \le \exp\Bigl(-\frac{N\Bigl(\frac{\beta_k}{B}\Bigr)^2}{3}\Bigr).
\]
We now ensure that the aforementioned failure probability is at most \(\frac{1}{N2^{k+1}}\), we have
\begin{equation*}
    \begin{split}
        &\Pr\Bigl\{X_{c+k+1} \le 2N\Bigl(\frac{\beta_k}{B}\Bigr)^2\Bigr\} \ge 1-\frac{1}{N2^{k+1}},\\[1mm]
        \Longrightarrow \quad &X_{c+k+1} \le 2N\left(\frac{\beta_k}{B}\right)^2 \equiv \beta_{k+1}.
    \end{split}
\end{equation*}

With the assumption that \(N = cB\). The previously proved recurrence exhibits a doubly exponential decay. In fact, one can show by induction that for all \(k\ge 0\), the following holds
\[
\beta_k \le B\cdot 2^{-\bigl(2^k-O(1)\bigr)}.
\]

Define \(\ell^{*}\) as the smallest integer (e.g., 0) or equivalently $\beta_{\ell^{*}} < \frac{1}{N}$, and $B\cdot 2^{-\bigl(2^{\ell^{*}}-O(1)\bigr)} < \frac{1}{N}$. Taking logarithms on both sides yields:
\begin{equation*}
    \begin{split}
        \log_2\Bigl(B\cdot 2^{-(2^{\ell^{*}}-O(1))}\Bigr) &< -\log_2 N,\\
        \log_2 B - (2^{\ell^{*}}-O(1)) &< -\log_2 N,\\
        2^{\ell^{*}}-O(1) &> \log_2 B + \log_2 N = \log_2 (BN) \\
        2^{\ell^{*}} & > 2\log_2 B + \log_2 c.\\
        \ell^{*} &> \log_2\log_2 B + O(1).
    \end{split}
\end{equation*}

Note that \(B = \Theta(N)\), and thus the aforementioned terms implies that when $\ell^{*}$ is larger than $O(\log\log N)$, the probability that $\exists \beta_{\ell^{*}}>0$ is at most $\frac{1}{N2^{k}}$. We then take the union bound over all $\beta_{k}$ where $k=0,1,..., \ell^{*}$, so that we can compute with probability at most $\sum_{k=0}^{\ell^{*}}\frac{1}{N2^k} < \frac{2}{N}$, we have $X_{c+\ell^{*}}=0$. Or in other word, with probability at most $1-\frac{1}{O(N)}$, the max bin load must be bounded by $c+O(\log_2\log_2 N)$.

\subsection{Proof of the tail bounds in Claim~\ref{claim:stash-sz}}\label{appx:tail}
We now define the excess above equilibrium $Y_t = X_t - x^{*}$ at each time $t$, and an exponential function $Z_t = e^{\lambda Y_t}$, where $\lambda > 0$ is a parameter to be optimized. Next, we show that $Z_t$ is a supermartingale for an appropriate choice of $\lambda$. 
For a bounded difference $|Y_{t+1} - Y_t| \leq c + \log_2\log_2 N$, we can use a standard inequality for the moment generating function:

\begin{equation*}
\begin{split}
\mathbb{E}[Z_{t+1} \mid Z_t] &= \mathbb{E}[e^{\lambda Y_{t+1}} \mid Z_t]\\
&= e^{\lambda Y_t} \cdot \mathbb{E}[e^{\lambda(Y_{t+1} - Y_t)} \mid Y_t]\\
&\leq e^{\lambda Y_t} \cdot \left(1 + \lambda\mathbb{E}[Y_{t+1} - Y_t \mid Y_t] + \frac{\lambda^2(c + \log_2\log_2 N)^2}{2}\right)\\
\end{split}
\end{equation*}

For $Y_t = \Delta > 0$, substituting the drift:
\begin{equation*}
\begin{split}
\mathbb{E}[Z_{t+1} \mid Z_t] &\leq e^{\lambda \Delta} \cdot \left(1 - \lambda\frac{2(1-\alpha)\Delta}{M} + \frac{\lambda^2(c + \log_2\log_2 N)^2}{2}\right)\\
\end{split}
\end{equation*}

For $Z_t$ to be a supermartingale, we need $\mathbb{E}[Z_{t+1} \mid Z_t] \leq Z_t = e^{\lambda \Delta}$. For this to hold for all $\Delta > 0$, we choose $\lambda = \frac{2(1-\alpha)\Delta}{M(c + \log_2\log_2 N)^2}$, then we compute the following inequalities:
With this choice of $\lambda$, $Z_t$ is a supermartingale. Using Markov's inequality:
\begin{equation*}
\begin{split}
\text{Pr}[X_t - x^{*} > \Delta] &= \text{Pr}[Y_t > \Delta]\\
&= \text{Pr}[Z_t > e^{\lambda \Delta}] \\
&\leq \frac{\mathbb{E}[Z_0]}{e^{\lambda\Delta}} \quad \text{(Markov's inequality)}\\
&= \frac{e^{\lambda(X_0 - x^{*})}}{e^{\lambda \Delta}} \\
&= e^{\lambda(X_0 - x^{*} - \Delta)} \\
&= e^{-\lambda(x^{*} + \Delta)} \quad \text{(Assuming $X_0 = 0$)}\\
&= e^{-\frac{2(1-\alpha)\Delta(x^{*} + \Delta)}{M(c + \log_2\log_2 N)^2}}\\
&= e^{-\frac{2(1-\alpha)\Delta x^{*}}{M(c + \log_2\log_2 N)^2} - \frac{2(1-\alpha)\Delta^2}{M(c + \log_2\log_2 N)^2}}\\
\end{split}
\end{equation*}

Substituting $x^{*} = \frac{(1+\alpha)M}{2}$:
\begin{equation*}
\begin{split}
\text{Pr}[X_t - x^{*} > \Delta] &= e^{-\frac{2(1-\alpha)\Delta \cdot \frac{(1+\alpha)M}{2}}{M(c + \log_2\log_2 N)^2} - \frac{2(1-\alpha)\Delta^2}{M(c + \log_2\log_2 N)^2}}\\
&= e^{-\frac{(1-\alpha)(1+\alpha)\Delta}{(c + \log_2\log_2 N)^2} - \frac{2(1-\alpha)\Delta^2}{M(c + \log_2\log_2 N)^2}}\\
&= e^{-\frac{(1-\alpha^2)\Delta}{(c + \log_2\log_2 N)^2} - \frac{2(1-\alpha)\Delta^2}{M(c + \log_2\log_2 N)^2}}\\
\end{split}
\end{equation*}

To establish a high probability bound, let us set:
\begin{equation*}
\Delta = \sqrt{\frac{M(c + \log_2\log_2 N)^2\ln N}{2(1-\alpha)}}
\end{equation*}

Substituting this value into our probability bound:
\begin{equation*}
\begin{split}
\text{Pr}[X_t - x^{*} > \Delta] &= e^{-\frac{(1-\alpha^2)\Delta}{(c + \log_2\log_2 N)^2} - \frac{2(1-\alpha)\Delta^2}{M(c + \log_2\log_2 N)^2}}\\
&= e^{-\frac{(1-\alpha^2)\Delta}{(c + \log_2\log_2 N)^2} - \frac{2(1-\alpha)}{M(c + \log_2\log_2 N)^2} \cdot \frac{M(c + \log_2\log_2 N)^2\ln N}{2(1-\alpha)}}\\
&= e^{-\frac{(1-\alpha^2)\Delta}{(c + \log_2\log_2 N)^2} - \ln N}\\
&= \frac{1}{N} \cdot e^{-\frac{(1-\alpha^2)}{(c + \log_2\log_2 N)^2} \cdot \sqrt{\frac{M(c + \log_2\log_2 N)^2\ln M}{2(1-\alpha)}}}\\
\end{split}
\end{equation*}

Since $(1-\alpha^2) > 0$ and all other terms are positive, the exponent is negative and grows with $\sqrt{M\ln M}$. Therefore, the above probability is at most $\frac{1}{O(N)}$. Or in other word, with high probability of at least $1-O\left(\frac{1}{M}\right)$, the queue size does not exceed:
\begin{equation*}
\frac{(1+\alpha)M}{2} + O\left((c + \log_2\log_2 N)\sqrt{M\ln N}\right)
\end{equation*}